\documentclass{llncs}

\usepackage{graphicx}
\usepackage{subfig}
\usepackage{mathtools}
\usepackage{amssymb}
\usepackage{amssymb,eepic,epic,latexsym}
\usepackage{cite}
\usepackage{color}

\newcommand{\abs}[1]{\left\lvert#1\right\rvert}

\title{On a Fire Fighter's Problem
}

\author{
	Rolf Klein$^1$
	\and Elmar Langetepe$^1$              
	\and Christos Levcopoulos$^2$ 
}

\institute{University of Bonn, Germany,
	Institute of Computer Science I.        
	\and 
	Lund University, Sweden,
	Department of Computer Science.
}

\date{\today}

\begin{document}
\maketitle

\begin{abstract}
Suppose that a circular fire spreads in the plane at unit speed. A single fire fighter can build a barrier at speed $v>1$. How large must $v$ be to ensure that the fire can be contained, and how should the fire fighter proceed? We contribute two results. 

First, we analyze the natural curve $\mbox{FF}_v$ that develops when the fighter keeps building, at speed $v$, a barrier along the boundary of the expanding fire. We prove that the behavior of this spiralling curve is governed by a complex function $(e^{w Z} -  s \, Z)^{-1}$, where $w$ and $s$ are real functions of $v$. For $v>v_c=2.6144 \ldots$ all zeroes are complex conjugate pairs.  If  $\phi$ denotes the complex argument of the conjugate pair nearest to the origin then, by residue calculus, the fire fighter needs $\Theta( 1/\phi)$ rounds before the fire is contained. As $v$ decreases towards $v_c$ these two zeroes merge into a real one, so that argument $\phi$ goes to~0. Thus, curve $\mbox{FF}_v$ does not contain the fire if the fighter moves at speed $v=v_c$.
(That speed $v>v_c$ is sufficient for containing the fire has been proposed before by Bressan et al.~\cite{bbfj-bsfcp-08}, who constructed a sequence of logarithmic spiral segments that stay strictly away from the fire.) 

Second, we show that any curve that visits the four coordinate half-axes in cyclic order, and in inreasing distances from the origin, needs speed $v>1.618\ldots$, the golden ratio, in order to contain the fire.

\end{abstract}

\section{Introduction}
Fighting wildfires and epidemics has become a serious issue in the last decades. Professional fire fighters need models and simulation tools
on which strategic decisions can be based; for example see \cite{faoun-ihffp}. Thus, a good understanding of the theoretical foundations seems necessary.

In theoretical computer science, substantial work has been done on the fire fighting problem in graphs; see, e.g., the survey article~\cite{fm-fpsr-09}. Here,
initially one vertex is on fire.  Then an immobile firefighter can be placed at one of the other vertices.
Next, the fire spreads to each adjacent vertex that is not defended by a fighter, and so on.
The game continues until the fire cannot spread anymore.
The objective, to save a maximum number of vertices from the fire, is NP-hard to achieve, even for trees of degree three; see~\cite{fkmr-fpgmd-07}. 
Optimal strategies are known for special graphs, i.e., 
for grid graphs~\cite{wm-fcg-02}.
The problem can also be interpreted as an 
intruder search game. The total extension of the fire in the graph represents 
the current possible location of the intruder. Some algorithms and lower bounds have been given for the problem of finding an
intruder in special graphs;
see\cite{bffs-cima-02,bggk-hmlnc-09,bkns-eosdi-07}.  

A more geometric setting has been studied in~\cite{kll-aagfb-14}. Suppose that inside a simple polygon $P$
a candidate set of pairwise disjoint diagonal barriers has been defined. If a fire starts at some point inside $P$ one wants to build a subset of these
barriers in order to save a maximum area from the fire. But each point on a barrier must be built before the fire
arrives there. This problem is a special case of a hybrid scheduling and coverage problem, for which an 11.65 approximation algorithm exists.

Bressan et al.~\cite{b-dicff-07,b-dbpmf-12,bbfj-bsfcp-08,bdl-eosfc-09,bw-msbhp-09,bw-efnaf-10,bw-gocdb-12,bw-osibp-12} introduced a purely geometric model where a fire spreads in the plane and one or more fire fighters are tasked to block it by building barriers. They cover a  wide range of possible scenarios, including fires whose shapes can change under the influence of wind, and provide a wealth of results, among them upper and lower bounds for the speed the fire fighter(s) need, and existence theorems for solutions that contain the fire and minimize the area burned. 

A very interesting case has been studied in Bressan et al.~\cite{bbfj-bsfcp-08}. A circular fire centered at the origin  spreads at unit speed, and a single fire fighter can build a barrier at speed~$v$. Their construction is based on the following observation. Let $P$ be a point on the logarithmic spiral $S^\alpha =(\varphi, e^{\varphi \cot\alpha})$ of excentricity $\alpha$, and let $Q$ denote the next point on $S^\alpha$ touched by a tangent at $P$; see Figure~\ref{logspiral-fig}, (i). 
\begin{figure}
\begin{center}
\includegraphics[scale=0.4]{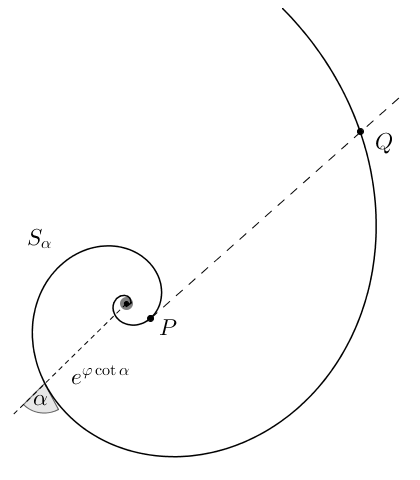}\includegraphics[scale=0.4]{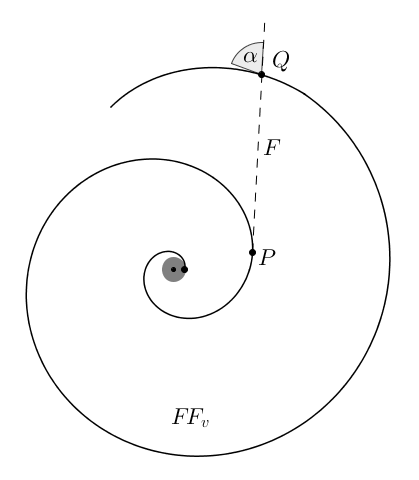}
\caption{(i) A logarithmic spiral $S_\alpha$ of excentricity $\alpha$; angle $\varphi$ ranges from $-\infty$ to $\infty$. (ii)~Curve $\mbox{FF}_v$ results when the fighter moves at speed $v$ along the fire's expanding boundary.}
\label{logspiral-fig}
\end{center}
\end{figure}
Then the spiral's length to $Q$ is at most $2.6144\ldots$ times the sum of its length to $P$ plus the length $|PQ|$ of the tangent, for all values of $\alpha$. In other words, if the fighter builds such a barrier at some speed $v>v_c:=2.6144\ldots$ she will always reach~$Q$ before the fire does, which crawls around the spiral's outside to point $P$ and then runs straight to $Q$. In~\cite{bbfj-bsfcp-08} the fighter uses this leeway to build a sequence of logarithmic spiral segments of increasing excentricities $\alpha_i$ that stay away from the fire, plus one final line segment that closes this barrier curve onto itself.

Logarithmic spiral movements have also been used 
for  shortest paths amidst growing circles~\cite{vo-psspa-06,mnrj-oacsp-07} and in the 
context of search games~\cite{l-ooss-10}.

In this paper we study the rather natural approach where the fire fighter keeps building a barrier right along the boundary of the expanding fire, at constant speed $v$. Let $\mbox{FF}_v$ denote the resulting barrier curve. At each point $Q$ both fighter and fire arrive simultaneously, by definition. As we shall see below, tangents form a constant angle $\alpha = \cos^{-1}(1/v)$ with the curve; see Figure~\ref{logspiral-fig}, (ii).

While the fighter keeps building $\mbox{FF}_v$, the fire is coming after her along the outside of the barrier, as shown in Figure~\ref{fire-fig}. 
Intuitively, the fighter can only win this race, and contain the fire, if the last coil of the barrier hits the previous coil.
\begin{figure}
\begin{center}
\includegraphics[scale=0.4]{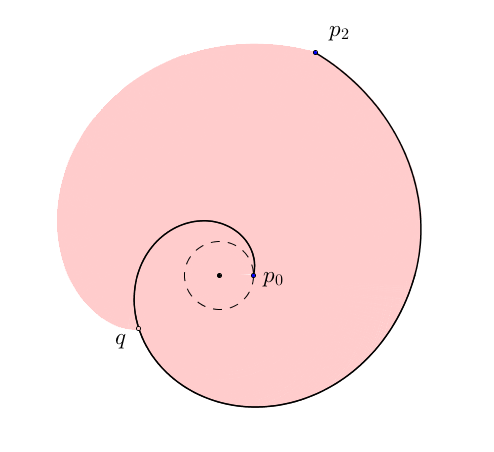}
\caption{The race between the fire and the fighter.
When the fighter arrives at point $p_2$, having constructed a barrier from $p_0$ to $p_2$, the fire has expanded along the outer side of the barrier up to point $q$.
}
\label{fire-fig}
\end{center}
\end{figure}

This is equivalent to saying that the length, $F$, of the tangent from the fighter's current position $Q$ back to point $P$  
becomes zero.

First, we deduce two structural properties of barrier curve $\mbox{FF}_v$. They lead to a recursive system of linear differential equations that allow us to describe the values of $F$ in the $i$-th round. It turns out that we need only check the signs of the values $F_i$ at the end of round~$i$, in order to see if the fighter is successful. 
Therefore, we look at the generating function $F(Z) = \sum_{i=0}^\infty {F_i Z^i}$ and obtain, from the recursions, the equation
\begin{eqnarray}
    \frac{F(Z)}{F_0} \ = \  \frac{ e^{v Z} \ - \ r \, Z}{ e^{w Z} \ - \ s \, Z}   \label{equa}
\end{eqnarray}
where $v, r, w, s$ are real functions of speed $v$. Singularities can only arise from zeroes of the denominator that do not cancel out with the numerator. It turns out that for $v > v_c = 2.6144\ldots$ only conjugate pairs of complex zeroes of the denominator exist~\cite{f-ansdd-95}. A theorem of Pringsheim's directly implies that not all coefficients $F_i$ of power series $F(Z)$ can be positive, showing that barrier curve $\mbox{FF}_v$ does close on itself at some time.

To find out after how many rounds this happens we look at the conjugate pair of zeroes of smallest modulus and let $\phi_v$ denote their (positive) argument. Residue analysis shows that  $\Theta(1/\phi_v)$ rounds are necessary before the fire is contained. As speed $v$ decreases towards $v_c$, the two conjugate zeroes merge into a real zero. Therefore, argument $\phi_v$ tends to~0, proving that the fire fighter cannot succeed at speed $v=v_c$.

\vspace{\baselineskip}
In addition to the results on curve $\mbox{FF}_v$, we obtain the following lower bound. Let us call a curve ``spiralling'' if it visits the four coordinate half-axes in cyclic order, and at increasing distances from the origin.
(Note that curve  $\mbox{FF}_v$ is spiralling even though the fighter's distance to~the origin may be decreasing: the barrier's intersection points with any ray from~$0$ are of increasing order since the curve does
not self-intersect.)

We prove that a fire fighter who follows such a spiralling curve can only be successful if her speed exceeds $\frac{1+\sqrt{5}}{2} \, \approx \, 1.618$, the golden ratio.

\section{Acknowledgement}
A preliminary version of part of this paper has appeared at SoCG'15~\cite{kll-ffp-15}. We thank all anonymous referees for their valuable suggestions and, in particular, for pointing out to us the work by Bressan et al.~\cite{b-dicff-07,b-dbpmf-12,bbfj-bsfcp-08,bdl-eosfc-09,bw-msbhp-09,bw-efnaf-10,bw-gocdb-12,bw-osibp-12}.

\section{The barrier curve $\mbox{FF}_v$} \label{FF-sec}
\subsection{The first rounds}

Let $p$ be a point on the barrier curve's first round, as depicted in Figure~\ref{Spirals-fig}. 
If $\alpha$ denotes the angle between the fighter's velocity vector at $p$ and the ray from~0 through $p$,  the fighter moves at speed $v \cos\alpha$ away from~0. This implies $v \cos\alpha =1$, because the fire expands at unit speed and the fighter stays on its frontier, by definition. Since the fighter is operating at constant speed $v$, angle $\alpha$ is constant, and given by $\alpha = \cos^{-1}(1/v)$.

Consequently, the first part of the barrier curve, between points $p_0$ and $p_1$ shown in Figure~\ref{Spirals-fig}, (i), is part of a logarithmic spiral of excentricity $\alpha$ centered at~0.
\begin{figure}
\begin{center}
\includegraphics[scale=0.29]{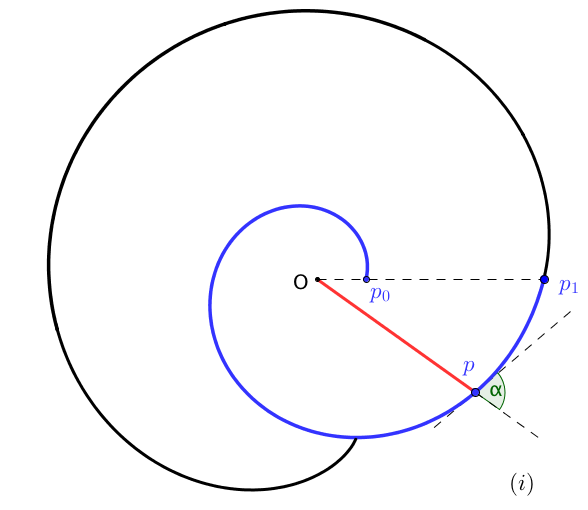}\includegraphics[scale=0.29]{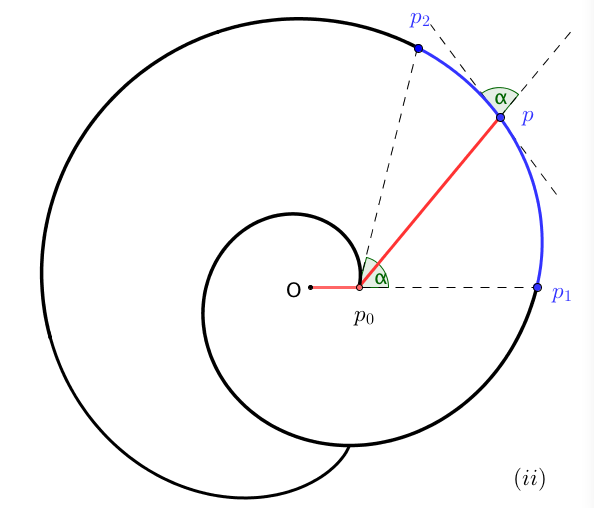}
\caption{The barrier curve starts with two parts of logarithmic spirals of excentricity $\alpha$, centered at 0 and $p_0$, respectively.}
\label{Spirals-fig}
\end{center}
\end{figure}
In polar coordinates, this segment can be desribed by $(\varphi, A\cdot e^{\varphi\cot\alpha})$, where $\varphi\in[0,2\pi]$, and $A$ denotes the distance from the origin to $p_0$, i.e., the fire's intitial radius.  

In general, the curve length of a logarithmic spiral of excentricity $\alpha$ between two points at distance $d_1<d_2$ to its center is known to be $\frac{1}{\cos\alpha}\left(d_2-d_1\right)$. Thus, we have for the length $l_1$ of the barrier curve from $p_0$ to $p_1$ the equation
\begin{eqnarray}
  l_1 \, = \, \frac{A}{\cos\alpha} \cdot (e^{2\pi \cot\alpha} - 1).    \label{l1-eq}
\end{eqnarray}

From point $p_1$ on, the geodesic shortest path, along which the fire spreads from 0 to the fighter's current position, $p$, is no longer straight. It starts with segment $0 p_0$, followed by segment $p_0 p$, until, for $p=p_2$, segment $p_0p$ becomes tangent to the barrier curve at $p_0$; see Figure~\ref{Spirals-fig}, (ii). 
By the same arguments as above, between $p_1$ and $p_2$  barrier curve $\mbox{FF}_v$ is also part of a logarithmic spiral of excentricity $\alpha$, but now centered at~$p_0$. 
This spiral segment starts at $p_1$ at distance 
 $A'=A(e^{2\pi \cot\alpha}-1)$ from its center $p_0$.
Since $p_2$ and $p_1$ form an angle $\alpha$ at $p_0$, the distance from $p_2$ to $p_0$ equals $A' e^{\alpha \cot\alpha}$. Thus, the curve length  from $p_1$ to $p_2$ is given by
 $l_2'=\frac{A'}{\cos\alpha} (e^{\alpha \cot\alpha}-1)=
 \frac{A}{\cos\alpha}  (e^{2\pi \cot\alpha} -1) (e^{\alpha \cot\alpha}-1)$.
 Consequently,  the overall 
curve length $l_2$ from $p_0$ to $p_2$ equals
\begin{eqnarray}
l_2 \, = \, l_1 + l'_2 \, = \, \frac{A}{\cos\alpha}  (e^{2\pi \cot\alpha} -1) e^{\alpha \cot\alpha}.   \label{l2-eq}
\end{eqnarray}%
\begin{figure}
\begin{center}
\includegraphics[scale=0.35]{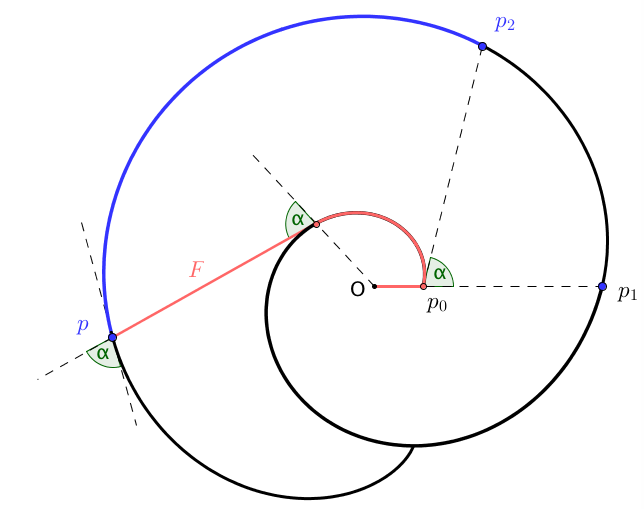}
\caption{From point $p_2$ on the barrier curve results from wrapping 
around the barrier already constructed. The last segment, \emph{free string} $F$, of the 
shortest path from the fire source to the current barrier point $p$ shrinks, 
by wrapping, and simultaneously grows by $\cos\alpha$. 
The fighter will be successful if, and only if,
$F$ ever shrinks to zero.}
\label{String-fig}
\end{center}
\end{figure}%
From point $p_2$ on, the geodesic shortest path from~0 to the fighter's current position, $p$, starts wrapping around the existing spiral part of the curve, beginning at $p_0$; see~Figure~\ref{String-fig}. 
The last segment of this path is tangent to the previous round of the curve. As mentioned in the Introduction, we shall endeavor to determine its length, $F$, because the fire will be contained if and only if $F$ ever attains the value~0.

One could think of this tangent as a string (named the {\em free string}) at whose endpoint,~$p$, a pencil is attached that draws the barrier curve. But unlike an involute, here the string is not normal to the outer layer. Rather, its extension beyond $p$ forms an angle~$\alpha$ with the barrier's tangent at $p$. This causes the string to grow in length by $\cos\alpha$ for each unit drawn. At the same time, the inner part of the string gets wrapped around the previous coil of the barrier. It is this interplay between growing and wrapping we need to analyze.

One can show that after $p_2$ no segment of positive length of $\mbox{FF}_v$ is part of a logarithmic spiral.

\subsection{Structural properties} \label{struct-subsec}

In this subsection we assume that the fighter has built quite a few rounds of the barrier curve without yet containing the fire. 
That the first two rounds of the curve involve two different spiral segments, around $0$ and around $p_0$, influences the subsequent layers. The structure of the curve can be described as follows.
Let $l_1$ and $l_2$ denote the curve lengths from $p_0$ to $p_1$ and $p_2$, respectively, as in Equations~\ref{l1-eq} and~\ref{l2-eq}. For $l \in [0,l_1]$ let $F_0(l)$ denote the segment connecting $0$ to the point of curve length $l$; see the sketch given in Figure~\ref{linkage-fig}. 
\begin{figure}
\begin{center}
\includegraphics[scale=0.5]{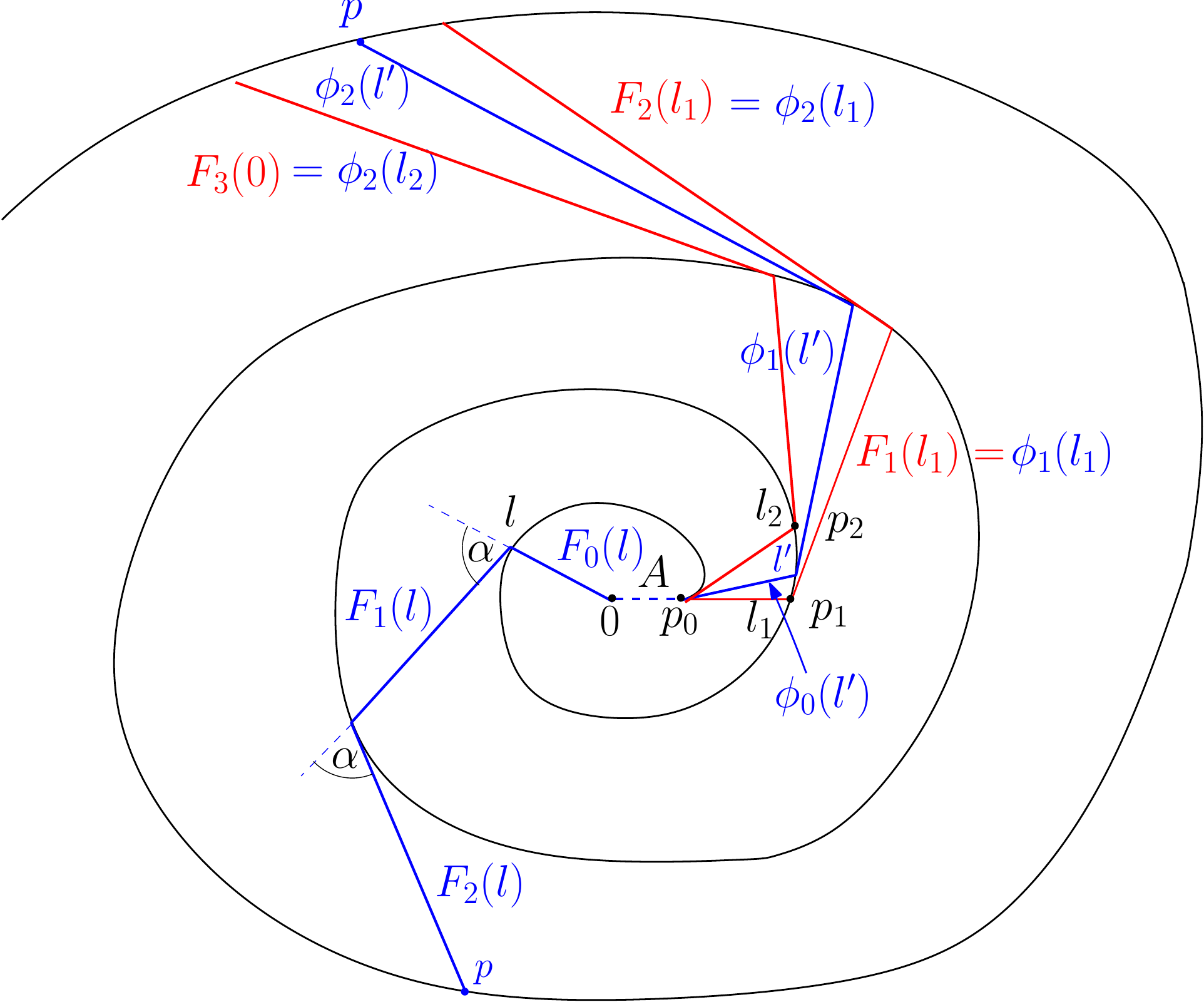}
\caption{Repeatedly constructing backwards tangents may end in 0 or in $p_0$. This way, two types of linkages are defined.}
\label{linkage-fig}
\end{center}
\end{figure}
At the endpoint of $F_0(l)$ we construct the tangent and extend it until it hits the next layer of the curve, creating a segment $F_1(l)$, and so on. This construction gives rise to a ``linkage''  connecting adjacent layers of the curve. Each edge of the linkage is turned counterclockwise by $\alpha$ with respect to its predecessor. The outermost edge of a linkage is the free string mentioned above.
As parameter $l$ increases from~$0$ to $l_1$, edge $F_0(l)$, and the whole linkage, rotate counterclockwise. While $F_0(0)$ equals the line segment from the origin to $p_0$, edge $F_0(l_1)$ equals segment $0p_1$.

Analogously, let $l' \in [l_1, l_2]$, and let $\phi_0(l')$ denote the segment from $p_0$ to the point
at curve length $l'$ from $p_0$. This segment can be extended into a linkage in the same way. We observe that 
\begin{eqnarray}
F_{j+1}(l_1) &=& \phi_{j+1}(l_1)  \label{l1}  \\
F_{j+1}(0) &= &\phi_{j}(l_2)    \label{l2} 
\end{eqnarray}
hold (but initially, we have $F_0(l) = A + \cos(\alpha) \, l$ and $\phi_0(l') = \cos(\alpha) \, l'$, so that
$F_0(l_1)\not= \phi_0 ( l_1)$).
Clearly, each point on the curve can be reached by a unique linkage, as tangents can be constructed backwards. We refer to the two types of linkages by $F$-type and $\phi$-type. As Figure~\ref{linkage-fig} illustrates, points of the same linkage type form alternating intervals along the barrier curve. If $p$'s linkage is of $F$-type then $p$ is uniquely determined by the index $j \geq 0$ and parameter $l \in [0,l_1] $ such that $p$ is the outer endpoint of edge $F_j(l)$.

\vspace{\baselineskip}
Now we will derive two structural properties of $F$-linkages on which our analysis will be based; analogous facts hold for $\phi$-linkages, too.
To this end, let $L_j(l)$ denote the length of the barrier curve from $p_0$ to the outer endpoint of edge $F_j(l)$, and let $F_j(l)$ also denote the length of edge $F_j(l)$.
\begin{lemma}     \label{grow-lem}
We have $L_{j-1}(l) + F_j(l) = \cos\alpha \, L_j(l)$.
\end{lemma}
\begin{proof}
Both, fire and fire fighter, reach the endpoint of $F_j(l)$ at the same time. The fire has travelled a geodesic distance of $L_{j-1}(l) + F_j(l)$ at unit speed, the fighter a distance of $L_j(l)$ at speed $1/\cos\alpha$.
\end{proof}

The second property is related to the wrapping of the free string. Intuitively, it says that if we turn an $F$-linkage, the speed of each edge's endpoint is proportional to its length.
\begin{lemma}       \label{links-lem}
As functions in $l$, $L_j$ and $F_j$ satisfy the following equation.
\[
    \frac{L'_{j-1}(l)}{L'_{j}(l)} \ = \ \frac{F_{j-1}(l)}{F_{j}(l)}.
\]
\end{lemma}

We will derive Lemma~\ref{links-lem} from two general facts on smooth curves stated in Lemma~\ref{wrap-lem} and Lemma~\ref{turn-lem}. 
\begin{lemma}   \label{wrap-lem}
Suppose a string of length $F$ is tangent to a point $t$ on some smooth curve $C$. Now the end of the string moves a distance of $\epsilon$ in the direction of $\alpha$, as shown in Figure~\ref{wrapApp-fig}. Then for the curve length ${C^{t_\epsilon}_t}$ between $t$ and the new tangent point, $t_\epsilon$, we have
\[
     \lim_{\epsilon \to 0}  \ \frac{{C^{t_\epsilon}_t}}{\epsilon} \ = \ \frac{\sin\alpha \ r}{F}
\]
where $r$ denotes the radius of the osculating circle at $t$.
\end{lemma}
This fact is quite intuitive. The more perpendicular the motion of the string's endpoint, and the larger the radius of curvature, the more of the string gets wrapped. But if the string is very long, the effect of the motion decreases.
\begin{figure}
\begin{center}
\includegraphics[scale=0.5]{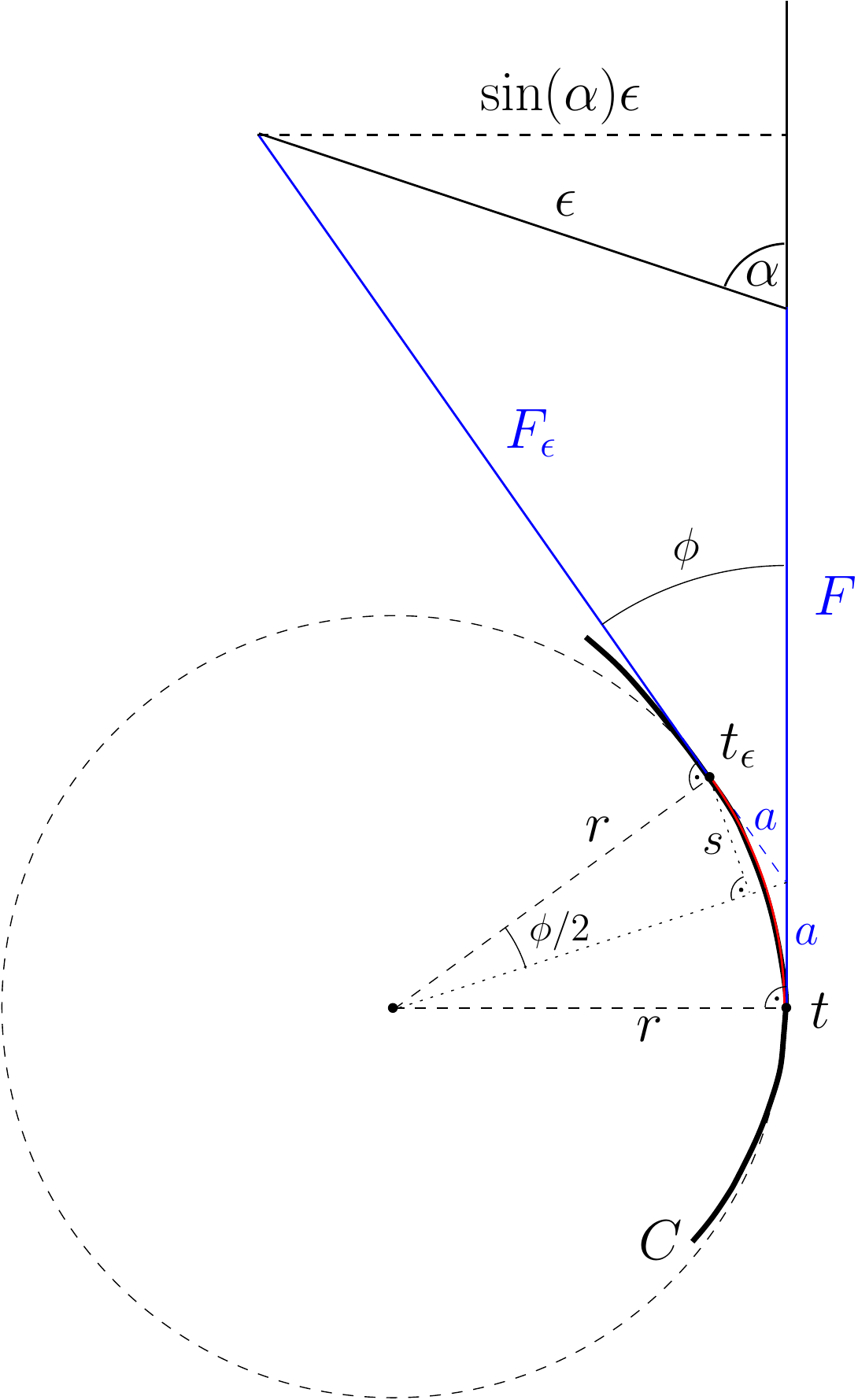}
\caption{A string wrapping around a curve.}
\label{wrapApp-fig}
\end{center}
\end{figure}

The center of the osculating circle at $t$ is known to be the limit of the intersections of the normals of all points near $t$ with the normal at $t$. Lemma~\ref{turn-lem} shows what happens if, instead of the normals, we consider the lines turned by the angle $\pi/2-\alpha$.
\begin{lemma}     \label{turn-lem}
Let $t$ be a point on a smooth curve $C$, whose osculating circle at $t$ is of radius~$r$. Consider the lines $L_{s}$ resulting from turning the normal at points $s$ by an angle of $\pi/2-\alpha$. Then their limit intersection point with $L_{t}$ has distance $\sin\alpha \, r$ to $t$.
\end{lemma} 
A simple example is shown in Figure~\ref{wrap-fig} for the case where curve $C$ itself is a circle. Now we can prove Lemma~\ref{links-lem}.
\begin{figure}
\begin{center}
\includegraphics[scale=0.65]{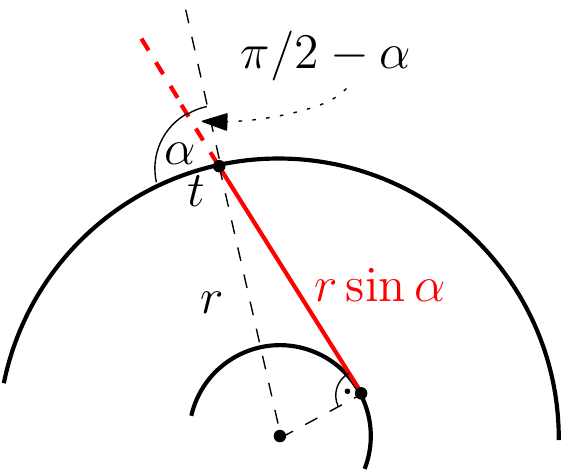}
\caption{Intersection of turned normals.}
\label{wrap-fig}
\end{center}
\end{figure}
\begin{proof}[Proof of Lemma~\ref{links-lem}]
By Lemma~\ref{wrap-lem}, applied to the innermost point $t$ of edge $F_j(L_j)$,
we have
\[
 \frac{L'_{j-1}(L_j)}{L'_{j}(L_j)} \, =  \, L'_{j-1}(L_j)  \, = \, \frac{\sin\alpha \ r}{F_j(L_j)}.
\]
Lemma~\ref{turn-lem} implies that $\sin\alpha \ r$ equals the distance between $t$ and the limit intersection point of the normals turned by $\pi/2 - \alpha$ near $t$. But for the barrier curve~$\mbox{FF}_v$, these turned normals are the tangents to the previous coil, so that $\sin\alpha \ r = F_{j-1}(L_j)$ holds. As we substitute variable $L_j$ with $L_j(l)$, the derivatives of the inner functions cancel out and we obtain Lemma~\ref{links-lem}.
\end{proof}

\begin{proof}[Proof of Lemma~\ref{wrap-lem}] 
Using the notations in Figure~\ref{wrapApp-fig}, the following hold.
From $r \, \sin(\phi/2) = s = a \, \cos(\phi/2)$ we obtain $a = r \, \tan(\phi/2)$. 
For short, let $c:=C^{t_\epsilon}_t.$
By l'Hospi\-tal's rule,
\[
    \frac{c}{2a} \ = \ \frac{r \, \phi}{2r \, \tan(\phi/2)} \ \approx \ \cos^2(\phi/2) \ \to \ 1
\]
as $\epsilon$, hence $\phi$, go to~$0$. Thus, $2a$ is a good approximation of $c={C^{t_\epsilon}_t}$. By the law of sines,
\[
      \frac{\epsilon \, \sin(\alpha)}{\sin(\phi)} \ = \ \frac{F_\epsilon + a}{\sin(\pi/2)},
\]
hence
\[
     \frac{\sin(\phi)}{\epsilon} \ = \ \frac{\sin(\alpha)}{F_\epsilon + a} \ \to \ \frac{\sin(\alpha)}{F}
\]
This implies $\sin(\phi/2)/\epsilon \ \to \ \sin(\alpha)/(2F)$,
and we conclude
\begin{eqnarray*}
\frac{C^{t_\epsilon}_t}{\epsilon} \ &=& \ \frac{c}{2a} \, \frac{2a}{\epsilon}  \ \approx \ \frac{2r \, \tan(\phi/2)}{\epsilon} 
=  \frac{2r \, \sin(\phi/2)}{\epsilon \, \cos(\phi/2)} \ \to \ \frac{r \sin(\alpha)}{F}.
\end{eqnarray*}~\end{proof}
\begin{proof}[Proof of Lemma~\ref{turn-lem}] 
Let us assume that $C$ is locally parameterized by $Y=f(X)$ and that $t=(x_0,f(x_0))$. Then the tangent in $t$ is
\[
Y \ = \ f'(x_0)  X \ - \ f'(x_0)  x_0 \ + \ f(x_0),
\]
and line $L_t$, the tangent turned counterclockwise by $\alpha$, is given by
\[
  Y \ = \ \tan(\arctan(f'(x_0)) + \alpha)  X \ - \   \tan(\arctan(f'(x_0)) + \alpha)  x_0 \ + \ f(x_0).
\]
Now let $(v,w)$ denote the point of intersection of $L_t$ and $L_s$, where $s=(x_0+\epsilon, f(x_0 + \epsilon))$.
Equating the two line equations we obtain
\[
   \big( h(x_0 +\epsilon) - h(x_0) \big) \, v \ = \ g(x_0 +\epsilon) - g(x_0) + f(x_0) - f(x_0+\epsilon)
\]
where
\[
   h(x) \, := \, \tan(\arctan(f'(x)) + \alpha) \ \mbox{ and } \ g(x) \, := \, h(x) x
\]
After dividing by $\epsilon$ and taking limits, we have
\[
    h'(x_0) \, v_0 = g'(x_0) - f'(x_0) = h'(x_0) \, x_0 + h(x_0) - f'(x_0),
\]
which results in
\begin{eqnarray*}
v_0 \ &=& \ x_0 \, + \, \frac{h(x_0)-f'(x_0)}{h'(x_0)}   \\
w_0 \ &=& \ h(x_0) \, v_0 -g(x_0) + f(x_0) \\
          &=& f(x_0) \, + \, \frac{h^2(x_0) - h(x_0) f'(x_0)}{h'(x_0)}.
\end{eqnarray*}
In other words,
\begin{eqnarray*}
(v_0,w_0) - (x_0, f(x_0)) \ &=& \ \frac{h(x_0) - f'(x_0)}{h'(x_0)} \, (1, h(x_0))   \\
|(v_0,w_0) - (x_0, f(x_0))| \ &=& \ | \frac{h(x_0) - f'(x_0)}{h'(x_0)} | \, \sqrt{1+h^2(x_0)}
\end{eqnarray*}
Using the addition formula for $\tan$,
\[
   h(x) \ = \ \tan(\arctan(f'(x)) + \alpha) \ = \ \frac{f'(x) + \tan(\alpha)}{1- f'(x) \tan(\alpha)},
\]
we obtain
\[
     h(x_0) - f'(x_0) \ = \ \frac{1 + (f'(x_0))^2 + \tan(\alpha)}{1- f'(x_0) \tan(\alpha)}.
\]
and
\[
   1+h^2(x_0) \ = \ \frac{( 1+(f'(x_0))^2) \, (1+\tan^2(\alpha))}{(1- f'(x_0) \tan(\alpha))^2}.
\]
Moreover,
\[
    h'(x_0) \ = \ \frac{f''(x)\, (1+\tan^2(\alpha)) }{(1- f'(x_0) \tan(\alpha))^2}.
\]
Putting expressions together we obtain
\[
   |(v_0,w_0) - (x_0, f(x_0))| \ = \ |\frac{\big( 1+(f'(x_0))^2 \big)^{3/2}}{f''(x_0)}| \ \frac{\tan(\alpha)}{\sqrt{1+\tan^2(\alpha)}}.
\]
The first term is known to be the radius of the osculating circle, $r$, and the second equals $\sin(\alpha)$.
\end{proof}
%

\section{Recursive differential equations}    \label{recdgl-sec}

In this section we turn the structural properties observed in Subsection~\ref{struct-subsec} into differential equations. 
By multiplication, Lemma~\ref{links-lem} generalizes to non-consecutive edges. Thus,
\begin{eqnarray}
   \frac{F_j(l)}{F_0(l)} \ = \ \frac{L'_j(l)}{l'} \ = \ L'_j(l)     \label{quots}
\end{eqnarray}
holds. On the other hand,  taking the derivative of the formula in Lemma~\ref{grow-lem} leads to
\begin{eqnarray}
     F'_j(l)  \ + \ L'_{j-1}(l) \ = \ \cos\alpha \,  L'_j(l)     \label{deriv}.
\end{eqnarray}
We substitute in~\ref{deriv} both $L'_j(l)$ and $L'_{j-1}(l)$ by the expressions we get from~\ref{quots} 
and obtain a linear differential equation for $F_j(l)$,
\begin{eqnarray*}     
   F'_j(l) \ - \ \frac{\cos(\alpha)}{F_0(l)} \, F_j(l) \ = \ - \, \frac{F_{j-1}(l)}{F_0(l)}\,. 
\end{eqnarray*}
The solution of $y'(x) + f(x)y(x) = g(x)$ is 
\[
   y(x)=\exp(-a(x)) \left(\int{g(t)\exp(a(t))} \,  \mathrm{d}t + \kappa \right),
\]
where $a = \int f$ and $\kappa$ denotes a constant that can be chosen arbitrarily. In our case, 
\[
   a(l) \, = \,  \int -\frac{ \cos(\alpha) }{A + \cos(\alpha) \, l} \, =  \, -  \ln(F_0(l))
\]
because of $F_0(l) = A + \cos(\alpha) \, l$,
and we obtain
\begin{eqnarray}
  F_j(l) \ = \   F_0(l)  \Big(  \kappa_j  \ - \ \int{  \frac{ F_{j-1}(t) }  {F^2_0(t)}   \,  \mathrm{d}t                                                      }    \Big).                  \label{Fdgl}
\end{eqnarray}
Next, we consider a linkage of $\phi$-type, for parameter $l \in [l_1, l_2]$, and obtain analogously
\begin{eqnarray}
  \phi_j(l) \ = \   \phi_0(l)  \Big(  \lambda_j  \ - \ \int{  \frac{ \phi_{j-1}(t) }  {\phi^2_0(t)}   \,  \mathrm{d}t }    \Big).                  \label{phidgl}
\end{eqnarray}
Now we determine the constants $\kappa_j, \lambda_j$ such that the solutions~\ref{Fdgl} and~\ref{phidgl} 
describe a contiguous curve. To this end, we must satisfy conditions~\ref{l1} and~\ref{l2}.
We define $\kappa_0 := 1$ and
\[
   \kappa_{j+1} := \frac{\phi_j(l_2)}{F_0(0)} \ + \ \int{ \frac{F_j(t)}{F^2_0(t)}  \mathrm{d}t }|_{l=0}
\]
so that~\ref{Fdgl}  becomes
\begin{eqnarray}
     F_{j+1}(l) \ = \  F_0(l) \, \Big(   \frac{\phi_j(l_2)}{F_0(0)}  \, - \,     \int_{0}^l {  \frac{ F_{j}(t) }{F^2_0(t)} }  \,  \mathrm{d}t     \Big),  \label{FRec}
\end{eqnarray}
which, for $l=0$,  yields $F_{j+1}(0) = \phi_j(l_2)$  (satisfying condition~\ref{l2}).

Similarly, we set $\lambda_0 :=1$ and
\[
   \lambda_{j+1} := \frac{F_{j+1}(l_1)}{\phi_0(l_1)} \ + \ \int{ \frac{\phi_j(t)}{\phi^2_0(t)}  \mathrm{d}t }|_{l=l_1}
\]
so that~\ref{phidgl}  becomes
\begin{eqnarray}
     \phi_{j+1}(l) \ = \  \phi_0(l) \, \Big(   \frac{F_{j+1}(l_1)}{\phi_0(l_1)}  \, - \,     \int_{l_1}^l {  \frac{ \phi_{j}(t) }{\phi^2_0(t)} }  \,  \mathrm{d}t     \Big),  \label{phiRec}
\end{eqnarray}
and for $l=l_1$ we get $F_{j+1}(l_1) = \phi_{j+1}(l_1)$  (satisfying condition~\ref{l1}).

\vspace{\baselineskip}
For simplicity, let us write
\begin{eqnarray}
     G_j(l) \, := \, \frac{F_j(l)}{F_0(l)} \ \mbox{ and } \ \chi_j(l) \, := \, \frac{\phi_j(l)}{\phi_0(l)},     \label{subst1}
\end{eqnarray}
which leads to
\begin{eqnarray}
     G_{j+1}(l) \ &=& \ \frac{\phi_0(l_2)}{F_0(0)} \, \chi_j(l_2) \ - \ \int_{0}^l {  \frac{ G_{j}(t) }{F_0(t)} }  \,  \mathrm{d}t  \label{G1}  \\
  \chi_{j+1}(l) \ &=& \ \frac{F_0(l_1)}{\phi_0(l_1)} \, G_{j+1}(l_1) \ - \ \int_{l_1}^l {  \frac{ \chi_{j}(t) }{\phi_0(t)} }  \,  \mathrm{d}t  \label{chi1}.   
\end{eqnarray}

The integrals in~\ref{G1} and~\ref{chi1} are increasing in $l$ provided that $G_j(t)>0$ and $\chi_j(t)>0$ hold for all $t$. 
This leads to a useful observation. 
In order to find out if the fire fighter is successful, 
(that is, if there exists an index $j$ such that $F_{j}(l) = 0$ holds for some $l \in [0,l_1]$, or $\phi_j(l)=0$ for some $l \in [l_1,l_2]$), 
we need to check only the values $F_j(l_1)$ at the end of each round. 
\begin{lemma}     \label{suff1-lem}
The curve encloses the fire if and only if there exists an index $j$ such that $F_{j}(l_1) \leq 0$ holds.
\end{lemma}
\begin{proof}
Clearly, $G_j$ and $F_j$ have identical signs, as well as $\chi_j$ and $\phi_j$ do. Suppose that $G_j >0$ and $G_{j+1}(l) =0$, for some $j$ and some
$l \in [0,l_1]$. By~\ref{G1}, function $G_{j+1}$ is decreasing, therefore $G_{j+1}(l_1) \leq 0$. Now assume that $G_i >0$ holds for all $i$, and that
we have $\chi_{j-1} >0$ and $\chi_{j}(l) =0$ for some $j$ and some $l \in [l_1,l_2]$. By~\ref{chi1} this implies $\chi_{j}(l_2) \leq 0$, and 
from~\ref{G1} we conclude  $G_{j+1} \leq 0$, in particular $G_{j+1}(l_1) \leq 0$.
\end{proof}

Next, we make the integrals in~\ref{G1} and~\ref{chi1} disappear by iterated substitution, and replace variable $l$ with the concrete values of $l_1$ (resp. $l_2$)
given in~\ref{l1-eq} and~\ref{l2-eq}. 

In equation~\ref{G1} iterated substitution yields
\begin{eqnarray}
    G_{j+1}(l) \ = \ \frac{\phi_0(l_2)}{F_0(0)} \, \sum_{\nu = 0}^j (-1)^\nu \,  I_\nu(l) \, \chi_{j-\nu}(l_2) \, + \, (-1)^{j+1} I_{j+1}(l)   \label{Gl}
\end{eqnarray}
where
\[
  I_n(x_n) \ = \  \int_{0}^{x_n} \frac{1}{F_0(x_{n-1})} \int_0^{x_{n-1}} \frac{1}{F_0(x_{n-2})} \ldots \int_0^{x_1} \frac{1}{F_0(x_0)} \, \mathrm{d}x_0 \ldots    \,  \mathrm{d}x_{n-1}.
\]
By induction on $n$ we derive
\[
    I_n(x_n) \ = \ \frac{1}{n!} \, \frac{1}{\cos^n\alpha} \, \big(\ln (\frac{A+\cos(\alpha) x_n}{A})\big)^n
\]
since $F_0(x)=A+\cos(\alpha) \, x$.  By definition of $l_1$, we have $\ln (\frac{A+\cos(\alpha) l_1}{A}) = 2 \pi \cot\alpha$,
so that setting $l=l_1$ in formula~\ref{Gl} leads to
\begin{eqnarray}
    G_{j}(l_1) \ = \ \frac{\phi_0(l_2)}{F_0(0)} \, \sum_{\nu = 0}^j{ \frac{(-1)^\nu}{\nu!} \,  \big( \frac{2 \pi}{\sin\alpha}  \big)^\nu \, \chi_{j-1-\nu}(l_2)}
\end{eqnarray}
where, for convenience, $\chi_{-1}(l_2):=\frac{F_0(0)}{\phi_0(l_2)}$. We observe that this formula is also true for $j=0$.
Multiplying both sides by $F_0(l_1)$, and re-substituting~\ref{subst1},  results in
\begin{eqnarray}
  F_{j}(l_1) \ = \ \frac{F_0(l_1)}{F_0(0)} \, \sum_{\nu = 0}^j{ \frac{(-1)^\nu}{\nu!} \,  \big( \frac{2 \pi}{\sin\alpha}  \big)^\nu \, \phi_{j-1-\nu}(l_2)}
                       \label{finF}
\end{eqnarray}
where $\phi_{-1}(l_2)=F_0(0)$.

\vspace{\baselineskip}
In a similar way we solve the recursion in~\ref{chi1}, using 
\[
   \int_{0}^{x_n} \frac{1}{\phi_0(x_{n-1})} \int_0^{x_{n-1}}  \ldots \int_0^{x_1} \frac{1}{\phi_0(x_0)} \, \mathrm{d}x_0 \ldots    \,  \mathrm{d}x_{n-1} \ = \ \frac{1}{n!} \, \frac{1}{\cos^n\alpha} \, \big(\ln (\frac{ x_n}{l_1})\big)^n
\]
and $\ln(\frac{l_2}{l_1}) = \alpha \cot\alpha$. One obtains, after substituting $l=l_2$,
\begin{eqnarray}
    \phi_{j}(l_2) \ = \ \frac{\phi_0(l_2)}{\phi_0(l_1)} \, \sum_{\nu = 0}^j{ \frac{(-1)^\nu}{\nu!} \,  \big( \frac{\alpha}{\sin\alpha}  \big)^\nu \, \hat{F}_{j-\nu}(l_1)}
                       \label{finphi}
\end{eqnarray}
where $\hat{F}_{0}(l_1):=\phi_0(l_1)$ and $\hat{F}_{i+1}(l_1):=F_{i+1}(l_1)$.

\section{Generating functions and singularities}

The cross-wise recursions~\ref{finF} and~\ref{finphi} are convolutions. In order to solve them
for the numbers $F_j(l_1)$ we are interested in, we define the generating functions
\[
F(Z) := \sum_{j=0}^\infty F_j \, Z^j\ \mbox{ and }  \  \phi(Z) := \sum_{j=0}^\infty \phi_j \, Z^j 
\]
where $F_j:=F_j(l_1)$ and $\phi_j:= \phi_{j}(l_2)$, for short. From~\ref{finF} we obtain 
\begin{eqnarray}
    F(Z) \ = \ \frac{F_0}{F_0(0)} \, e^{-\frac{2 \pi}{\sin\alpha} Z} \, \big( Z \, \phi(Z) + F_0(0)  \big),    \label{Fsum1}
\end{eqnarray}
and from~\ref{finphi},
 \begin{eqnarray}
    \phi(Z) \ = \ \frac{\phi_0}{\phi_0(l_1)} \, e^{-\frac{\alpha}{\sin\alpha} Z} \, \big( Z \, F(Z) - F_0 + \phi_0(l_1)  \big)    \label{phisum1}.
\end{eqnarray}
Both equalities can be verified by plugging in expansions of the exponential functions, using $e^W=\sum_{j=0}^\infty {\frac{W^j}{j!}}$, computing the products, and comparing coefficients. Now we substitute~\ref{phisum1} into~\ref{Fsum1},
solve for $F(Z)$, divide both sides by $F_0$ and expand by $e^{\frac{2 \pi + \alpha}{\sin\alpha}}$ to obtain 
\begin{eqnarray}
  \frac{F(Z)}{F_0} \ = \  \frac{ e^{v Z} \ - \ r \, Z}{ e^{w Z} \ - \ s \, Z},    \label{efunc}
\end{eqnarray}
where $v, r, w, s$ are the following functions of $\alpha = \cos^{-1}(1/v)$:
\begin{eqnarray}
   v \ &=& \  \frac{\alpha}{\sin\alpha} \ \ \ \  \, \mbox{ and } \ r \ = \ e^{\alpha \cot\alpha}\nonumber   \\
   w \ &=& \ \frac{2 \pi +\alpha}{\sin\alpha} \ \  \mbox{ and } \ s \ = \ e^{(2 \pi + \alpha) \cot\alpha}. \label{sw}
\end{eqnarray}

Singularities of $F(Z)$ can arise only from zeroes of the denominator,
\begin{eqnarray}     \label{denominator}
e^{w Z}  -  s  Z. 
\end{eqnarray}
As the fighter's speed $v$ increases from 1 to $\infty$, angle $\alpha=\cos^{-1}(1/v)$ of the fighter's velocity vector increases from 0 to $\pi/2$, causing $s/w$ to decrease from $\infty$ to 0. Precisely at $v=v_c=2.6144\ldots$ does $s/w=e$ hold; then $1/w$ is a real root of~\ref{denominator}, as direct computation shows. For $v>v_c$ we have $s/w < e$.

The following lemma can be inferred from Falbo~\cite{f-ansdd-95}. A complete proof is given in Subsection~\ref{Falbo-subsec} below.
\begin{lemma}   \label{FalboApp-lem}
For $v>v_c$ function $F(Z)$ has infinitely many discrete conjugate complex poles of order~1. The pair $z_0, \overline{z_0}$ nearest to the origin have absolut values $< 0.31$, all other poles have moduli $\geq 1$. As $v$ decreases to $v_c$, poles $z_0, \overline{z_0}$ converge towards a real pole $1/w \approx 0.124$.
\end{lemma}

Now we directly obtain the following.

\begin{theorem}   \label{qual-theo}
At speed $v>v_c=2.6144\ldots$ curve $\mbox{FF}_v$ contains the fire.
\end{theorem}
\begin{proof}
From Lemma~\ref{FalboApp-lem} we know that $F(Z)$ has a radius of convergence $R$ in $(0,0.31)$. If the fighter were unsuccessful then
all coefficients $F_j$ of $F(Z)$ would be positive. By a theorem of Pringsheim's (see, e.g.,~\cite{fs-ac-09} p.~240), this would imply that $R$ is a singularity of $F(Z)$. But there are only complex singularities, due to Lemma~\ref{FalboApp-lem}.
\end{proof}

\subsection{Proof of Lemma~\ref{FalboApp-lem}}     \label{Falbo-subsec}
The equation $e^{wZ} - sZ  =  0$ has received some attention in the field of delay differential equations, 
see, e.g., Falbo~\cite{f-ansdd-95}. 
With the following claim our main interest
will be in case~(i) and its transition into case~(ii).

\noindent
{\bf Claim 1}   \label{zero-lem}
$\mbox{}$\\
(i) If $\frac{s}{w} < e$ then the equation $e^{wZ} - sZ  =  0$ has an infinite number of non-real, 
discrete conjugate pairs of complex roots. \\
(ii) As $\frac{s}{w}$ increases to $e$, the pair of complex roots $z_0$ and $\overline{z_0}$ of minimum imaginary part
converge to the real zero $x_0=1/w$.\\
(iii) For $\frac{s}{w} > e$, the real zero $x_0$ splits into two different real zeros.

\begin{proof}
Let $z=a+i \, b$ be a complex zero of $e^{wZ} - sZ \, = \, 0$, for real parameters $w,s \not= 0$, that is,
\begin{eqnarray}
     e^{wa} \, \big( \cos(wb) + i\, \sin(wb) \big) \, = \, sa+i \, sb.   \label{zero}
\end{eqnarray}
If the imaginary part $b$ of $z$ equals zero then $e^{wa}=sa$, hence
\[
     \frac{e^{wa}}{wa} \, = \, \frac{s}{w}. 
\]
This implies $\frac{s}{w} \geq e$; see Figure~\ref{exdx-fig}.
\begin{figure}
\begin{center}
\includegraphics[scale=0.4]{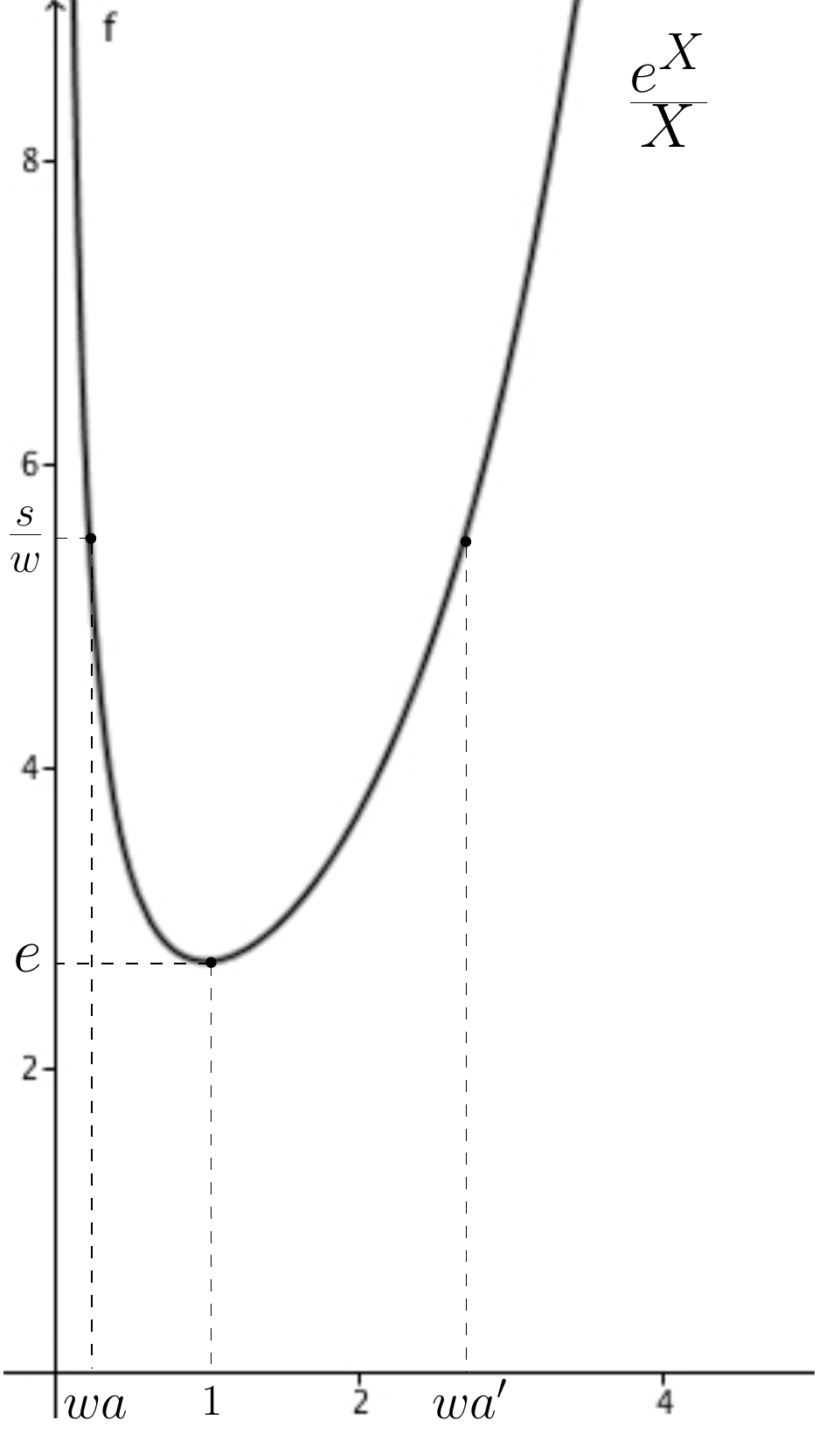}
\caption{The real function $e^X/X$ takes on only values $\geq e$, and those $>e$ exactly twice.}
\label{exdx-fig}
\end{center}
\end{figure}
Now suppose that $b\not=0$ holds. Then~\ref{zero} implies 
\begin{eqnarray*}
   e^{wa} \, \cos(wb) \, &=& \, sa \\
  e^{wa} \, \sin(wb) \, &=& \, sb ,
\end{eqnarray*}
hence $\cot(wb) = \frac{a}{b}$ and
\begin{eqnarray*}
   e^{wb \, \cot(wb)} \, \sin(wb) \, = \, \frac{s}{w} \, wb.     \label{imfunc} 
\end{eqnarray*}
The graph of the real function $h(X):=e^{X \cot X} \sin X$ intersects the line $q \, X$ in an infinite number of discrete points;
see Figure~\ref{ke-fig}.
\begin{figure}
\begin{center}
\includegraphics[scale=0.4]{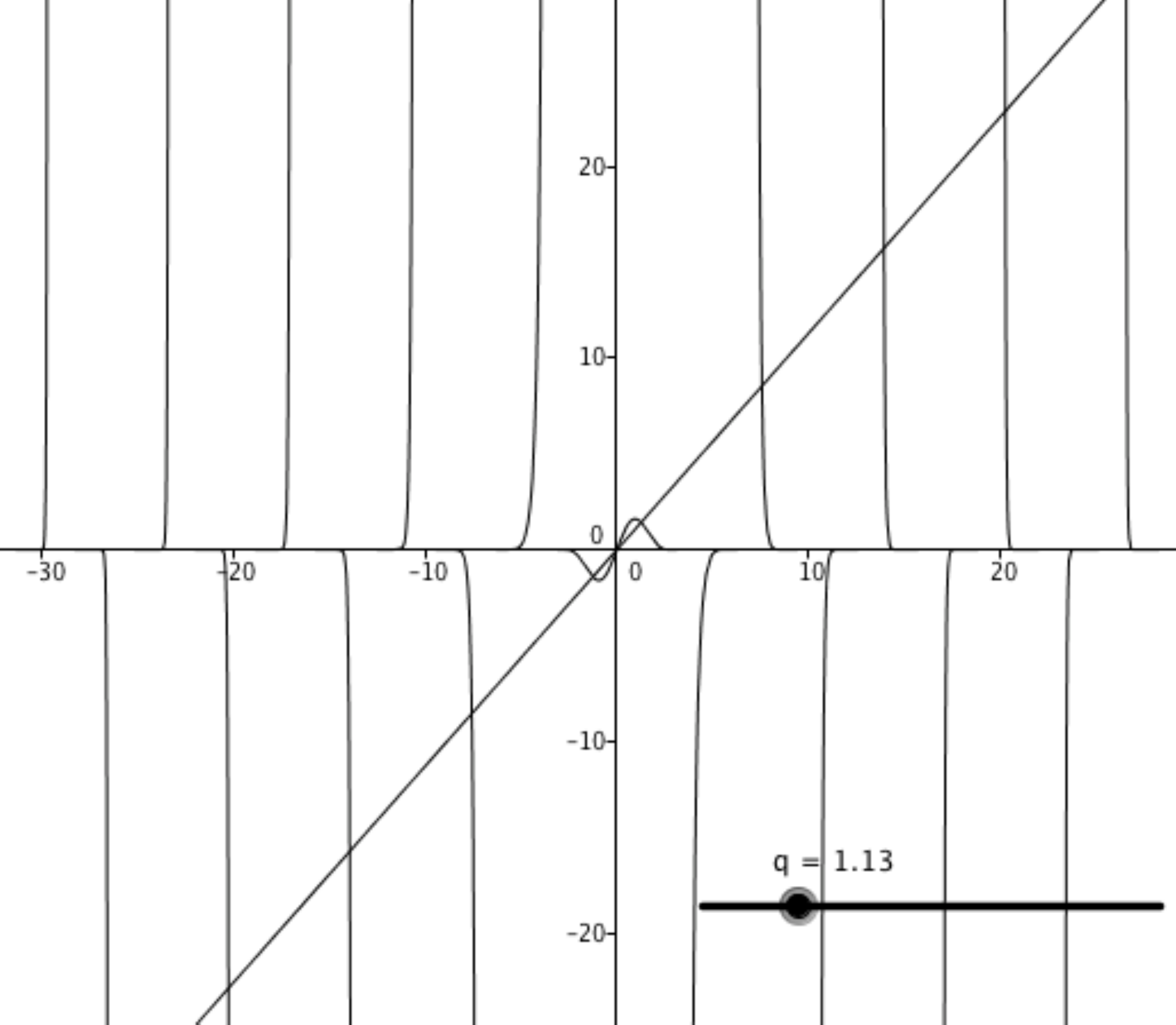}
\caption{As slope $q$ of line $qX$ grows to $e$, its intersections with the center part of $e^{X \cot X} \sin X$ disappear.}
\label{ke-fig}
\end{center}
\end{figure}
Each intersection point $p$ with abscissa $x$ corresponds to a zero $\cot(x)\frac{x}{w}+\frac{x}{w} \, i$  
of $e^{wZ} - sZ  =  0$, of absolute value 
$\frac{1}{\sin^2 x} \,  \frac{x^2}{w^2}$.
Function $h(X)$ has poles at the integer multiples of $\pi$. As shown in Figure~\ref{poles-fig},
the first intersection point to the right of $0$ has abscissa 
$x_0<\pi$, the following ones,  $x_k>2k\pi$.

As slope $q$ of the line $qX$ increases beyond $e$, 
its two innnermost intersections $\not=0$ with the graph of $h(X)$ disappear.
Thus, the imaginary parts of  $z_0, \overline{z_0}$ become zero,
causing a ``double'' real zero at $1/w$. As $q$ grows beyond $e$, this zero splits
into two simple zeroes $a/w$ and $a'/w$; compare Figure~\ref{exdx-fig}.

For later use we note the following. While slope $q$ is less than $e$ we can write the zero $z_0$
of positive imaginary part associated with $p_0$ as 
\begin{eqnarray}
     z_0 \, = \,  a + b \, i \, = \, \rho \, (\cos(\phi) + \sin(\phi)) \, i.     \label{z0}
\end{eqnarray}
This representation yields $\cot(\phi) = a/b$. Since we have also derived $\cot(wb) = a/b$ it follows that angle $\phi$
of $z_0$ and $wb$ are congruent modulo $\pi$. Since $wb=x_0 < \pi$ and $\phi < \pi$ because of $b>0$
we conclude that
\begin{eqnarray}
     \phi \, = \, wb \, = \, w \rho \sin(\phi)     \label{phival}
\end{eqnarray}
is the smallest positive solution of $e^{X \cot X} \sin X = qX$.
\begin{figure}
\begin{center}
\includegraphics[scale=0.4]{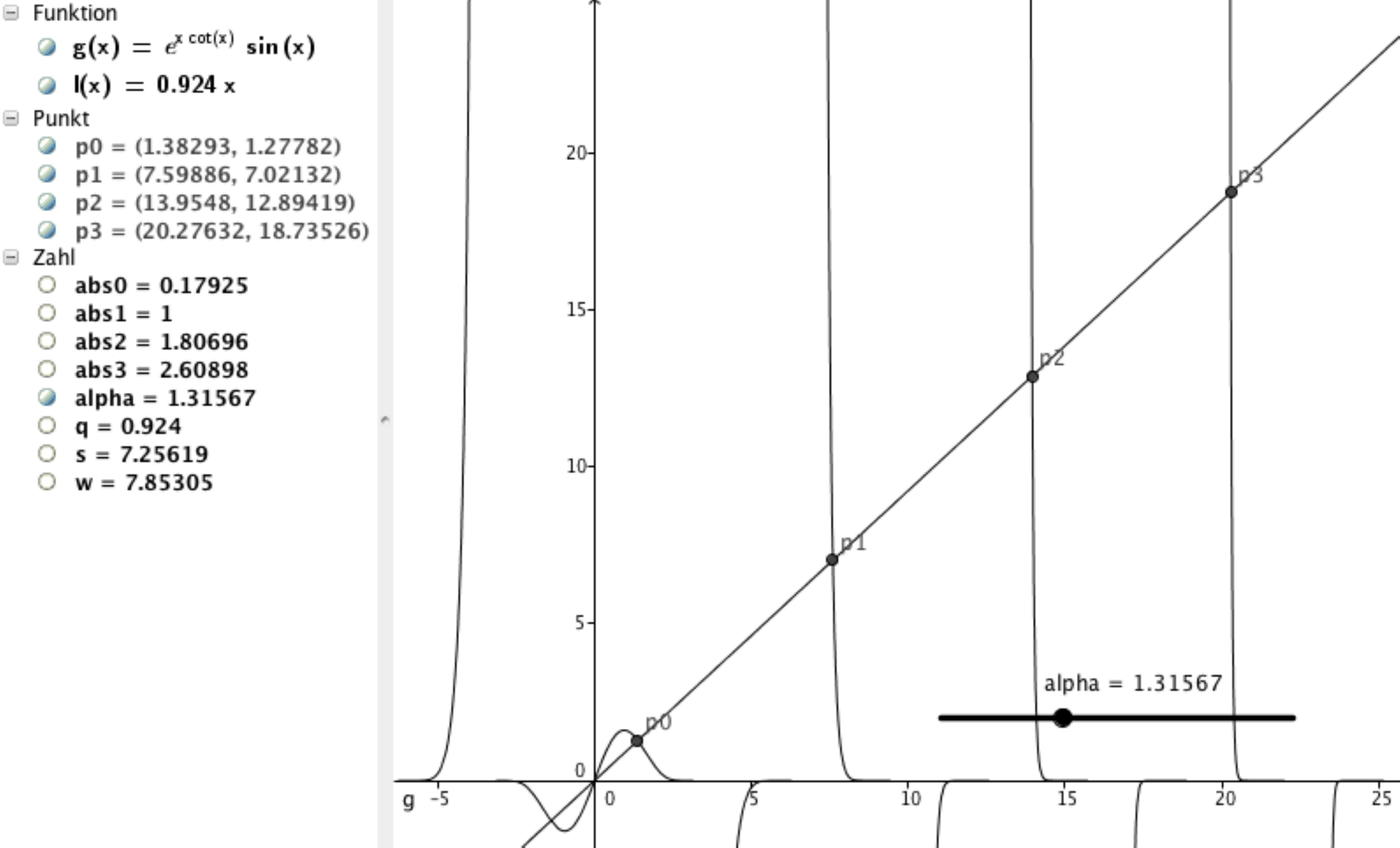}
\caption{The first intersection points of $h(X):=e^{X \cot X} \sin X$ with line $qX$ to the right of $0$. 
Here, $q=s/w$ is a function of angle $\alpha$. The numbers $\mbox{abs}i$ denote the absolute values of the zeroes
of $e^{wZ} - sZ  =  0$ that correspond to the intersection points $p_i$, for $i=0,\ldots,3$. The value of $\mbox{abs}0$
is decreasing towards $0.30563$, as $\alpha$ tends to $\pi/2$. We have $\mbox{abs}1=1$ because of the zero in
Lemma~\ref{common-lem}~(i). All other zeroes have absolute values substantially larger than~1.}
\label{poles-fig}
\end{center}
\end{figure}

\end{proof}

First we show that only poles can arise from these zeroes.

\noindent
{\bf Claim 2}      \label{pole-lem}
$\mbox{}$\\
Each zero $u=a+b \, i$ (except $1/w$ in case~(ii) of Claim 1) is a pole 
of order one of the function 
\[
d(Z)\, :=\, \frac{1}{e^{w Z} - s \, Z} 
\]
with residue $\mu=((wu-1)s)^{-1}$.

\begin{proof}
We have 
\[
\frac{Z-u}{e^{wZ}-sZ} \, = \, \frac{Z-u}{e^{wZ}-e^{wu}+su-sZ} \, = \, \frac{1}{\frac{e^{wZ}-e^{wu}}{Z-u}-s}.
\]
As $Z$ tends to $u$, the differential quotient in the denominator tends to the finite number $(e^{wZ})'(u) = we^{wu} = wsu$.
Hence,  $u$ is a pole of order~1, and $d(Z)$ has local expansion
\[
   d(Z) \ = \ \frac{\mu}{Z-u} \, + \, \sum_{i=0}^\infty w_i \, (Z-u)^i.
\]
\end{proof}
One can show that, in case~(ii), zero $1/w$ gives rise to a pole of order two of $d(Z)$.
Thus, function $d(Z)$, and therefore 
\begin{eqnarray}
    f(Z) \ := \  \frac{ e^{v Z} \ - \ r \, Z}{ e^{w Z} \ - \ s \, Z}   \label{cfunc}
\end{eqnarray}
are meromorphic.

Next, we consider which of the poles of $d(Z)$ cancel out in the numerator of $f(Z)$. From now on, 
the parameters $v,r,w,s$ are no longer considered independent but functions of angle~$\alpha$.

\noindent
{\bf Claim 3}         \label{common-lem}
$\mbox{}$\\
Numerator and denominator of function $f(Z)$ in~\ref{cfunc} have the following zeroes in common:\\
(i) $\cos(\alpha) + \sin(\alpha) \, i$   \\
(ii) $\cos(\alpha) + (q+1)\sin(\alpha) \, i$ for each integer $q$ satisfying $\alpha=\frac{2p}{q} \pi$, for some integer $p$.\\

\begin{proof}
Let $z=a+b \,i$ be a common zero of $e^{vZ}-rZ$ and $e^{wZ}-sZ$. As in the proof of Claim~1
we have
\begin{eqnarray}
   e^{wb \, \cot(wb)} \, \sin(wb) \, &=& \, sb     \\
   e^{wb \, \cot(wb)} \, \cos(wb) \, &=& \, sa          \label{s}
\end{eqnarray}
and, analogously,
\begin{eqnarray}
   e^{vb \, \cot(vb)} \, \sin(vb) \, &=& \, rb     \\
   e^{vb \, \cot(vb)} \, \cos(vb) \, &=& \, ra.         \label{r}
\end{eqnarray}
This implies $\cot(wb)=a/b=\cot(vb)$, hence $wb=vb+k\pi$ for some integer~$k$.
Because $s$ and $r$ are positive for all $\alpha$, the expressions in~\ref{s} and~\ref{r}
must have the same sign, and we conclude that $k=2h$ is even. This implies 
$\sin(wb)=\sin(vb)$ and from
\[
    \frac{2\pi}{\sin\alpha} b = (w-v)b =2h\pi
\]
follows $b=h\sin\alpha$.
 Moreover, we have
\[
   e^{\frac{2\pi}{\sin\alpha} a} = e^{(w-v)a} = \frac{e^{wa}\sin(wb)}{e^{va}\sin(vb)} = \frac{sb}{rb} =\frac{s}{r} = e^{2\pi \cot\alpha},
\]
and we obtain $a=\cos\alpha$. This yields
\[
r \cos(vb) = e^{v \cos\alpha}\cos(vb) = e^{va}\cos(vb) = ra = r\cos\alpha,
\]
hence $vb=\alpha + 2p\pi$ for some integer $p$, and from 
$h\alpha = hv\sin\alpha=vb$ follows
\[
     \alpha \, = \, \frac{2p}{h-1} \pi.
\]
\end{proof}

From now on we consider only case~(i) of Claim~1. 
Since $\frac{s}{w}$ is a strictly decreasing function in $\alpha$,  we have
\begin{eqnarray}
\frac{e^{(2 \pi + \alpha) \cot\alpha}}{\frac{2 \pi + \alpha}{\sin\alpha}} \ =  \ \frac{s}{w} \ &<& \ e           \label{alfac} \\
     \Longleftrightarrow \ \alpha \, &>& \, \alpha_c \, := \, 1.17830\ldots.
\end{eqnarray}
The critical angle $\alpha_c$ corresponds to a speed $v_c = 1/\cos(\alpha_c) = 2.61440\ldots$. 
For $\alpha \in (\alpha_c, \pi/2)$ we can summarize our findings as follows. 

\noindent
{\bf Claim 4}    \label{summa-lem}
$\mbox{}$\\
For $\alpha \in (\alpha_c, \pi/2)$, function $f(Z)$ in~\ref{cfunc} has only non-real, first-order poles as singularities. A conjugate pair $z_0, \overline{z_0}$ is situated at distance $<0.31$ from the origin. All other poles are of absolute value $>1$. For $\alpha \rightarrow \alpha_c$ both $z_0,\overline{z_0}$ converge to the real pole $(1/w,0)$ where $1/w  \approx 0.12383.$ 

\begin{proof} 
By Claim~2, and by Claim~1~(i), 
function $f(Z)$ has only poles of order one for singularities, none of which are real. Zero $\cos\alpha+ \sin\alpha \, i$
of the denominator is of absolute value~1, but it is not a pole of $f(Z)$, by Claim~3~(i). All other zeroes of the denominator canceling out must be of absolute value $>1$, by Claim~3~(ii). Hence, $z_0$ and $\overline{z_0}$
are in fact poles of $f(Z)$. The bounds on the absolute values can be obtained by numerical evaluation; see Figure~\ref{poles-fig}.
\end{proof}
This concludes the proof of Lemma~\ref{FalboApp-lem}.

\section{Residue analysis}
From now on let $v>v_c$.
In order to find out how many rounds barrier curve $\mbox{FF}_v$ runs before it closes down on itself we employ a technique of Flajolet's; see~\cite{fs-ac-09} p.~258 ff. Instead of $F(Z)=F_0\, f(Z)$ we consider the function 
\begin{eqnarray}
   g(Z) \ &:=& \   \frac{1}{Z^{j+1}} \, \frac{F(Z)}{F_0}     \ = \  \frac{1}{Z^{j+1}} \,  \frac{ e^{v Z} \ - \ r \, Z}{ e^{w Z} \ - \ s \, Z}  \\
        \ &=& \ \frac{\frac{F_0}{F_0}}{Z^{j+1}} \, + \, \frac{\frac{F_1}{F_0}}{Z^{j}} \, + \, \frac{\frac{F_2}{F_0}}{Z^{j-1}}
                                         + \, \ldots \, \frac{{\color{blue} \frac{F_j}{F_0}}}{Z} \ + \ \sum_{i=0}^\infty \frac{F_{j+i+1}}{F_0} \, Z^i   \label{zeropole}
\end{eqnarray}
By Lemma~\ref{FalboApp-lem}, function $g(Z)$ has complex poles at $z_0$ and $\overline{z_0}$, and a real pole at the origin, and these are the only singularities inside the circle $\Gamma$ of radius~$\gamma:=0.9$; see Figure~\ref{poles3-fig}. Let $\mu$ and $\overline{\mu}$ denote the residues of the poles at $z_0$ and $\overline{z_0}$. By~\ref{zeropole}, the pole at~0 has residue $F_j/F_0$, the coefficient we are interested in.
\begin{figure}
\begin{center}
\includegraphics[scale=0.6]{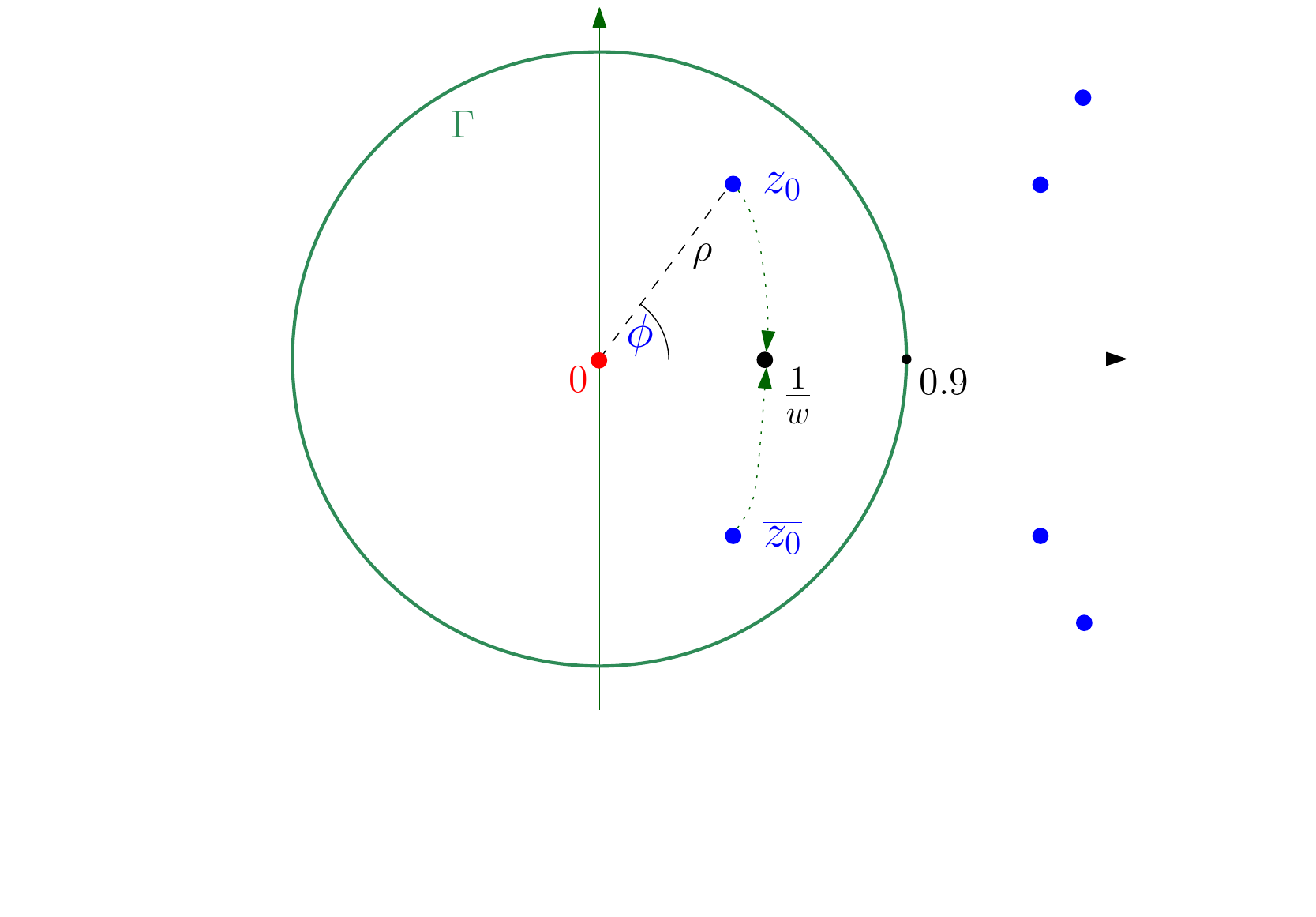}
\caption{The poles of minimum modulus of $F(Z)$.}
\label{poles3-fig}
\end{center}
\end{figure}
By Cauchy's Residue Theorem, 
\begin{eqnarray}
\frac{F_j}{F_0} \, &=& \, - (\mu + \overline{\mu}) \ + \  \frac{1}{2 \pi \, i} \, \int_\Gamma g(u) \, \mathrm{d}u .   \label{CRT}
\end{eqnarray}
%

The integral can be upper bounded by
\begin{eqnarray}
    \abs{\frac{1}{2 \pi \, i} \, \int_\Gamma\frac{f(z)}{z^{j+1}}  \, \mathrm{d}z} \ \leq \ \frac{1}{2 \pi} \, 
    D \,    \int_\Gamma\frac{1}{\abs{z^{j+1}}}  \, \mathrm{d}z \ &=& \ \frac{1}{2 \pi} \, D \, \Big(\frac{1}{\gamma}\Big)^{j+1} \, 2 \pi \, \gamma   \\
                            &=& \ D \, \gamma^{-j}    \label{int}
\end{eqnarray}  
because all $z$ on $\Gamma$ have absolute value $\gamma=0.9$. Here, $D$ denotes the maximum value function $\abs{f(z)}$ attains on the compact set $\Gamma \times [\alpha_c,\pi/2]$.

Now let
\begin{eqnarray}
    z_0 \, = \, \rho (\cos\phi + \sin\phi \, i)  \, = \, a + b \, i                   \label{coord}
\end{eqnarray}
be the pole different from zero whose imaginary part, $b$, is positive. 
Let us recall from~\ref{phival}
in the proof of Claim~1 that $\phi$, the angle of $z_0$, is the smallest positive real number
solving $e^{x \cot(x)} \sin(x) = \frac{s}{w} x$, with $s, w$ as defined in~\ref{sw}. We also have shown that
\begin{eqnarray}
    \phi \, = \, w b \, = \, w \rho \, sin\phi     \label{philab}
\end{eqnarray}
holds. Furthermore, $\rho=\abs{z_0}$. Let $x_0=(1/w,0)$ denote the real pole to which $z_0$ and $\overline{z_0}$  converge as 
$\alpha$ decreases to $\alpha_c$; compare Claim~4. We obtain the following identities.
\begin{eqnarray}
\ w \abs{z_0-x_0} \ &=& \ \sqrt{w^2 \rho^2 \, - \, 2w \rho \cos(\phi) +1}     \\
                                &=& \ \sqrt{\frac{\phi^2}{\sin^2\phi} -2\phi \frac{\cos\phi}{\sin\phi} +1}  \\
                                &=& \ \sqrt{(wa-1)^2+w^2b^2}   \label{wurz}
\end{eqnarray}

The sum of residues can be written as follows.
\begin{lemma}      \label{ressum-lem}
We have
\begin{eqnarray}
- (\mu + \overline{\mu}) \ = \ \Big(\  \frac{e^{va}}{\rho^2}\, \cos((j+1)\phi - vb) \, &-& \, \frac{e^{va}w}{\rho}\,      \label{major}
                                                                \cos((j+2)\phi -vb)    \\
                                                          - \ \big(  \ \frac{r}{\rho} \, \cos(j \phi) \, &-&  \ rw \, \cos((j+1)\phi)   \ \big) \  \Big)   \label{minor}\\ 
\ \cdot \  \frac{2}{s\big( (wa-1)^2+w^2b^2 \big)} \, &\cdot& \, \frac{1}{\rho^{j-1}}      \label{power}
\end{eqnarray}
\end{lemma}  
%
\begin{proof}
Using $z_0=a + b \, i$ we obtain
\begin{eqnarray*}
\mu + \overline{\mu} \ &=& \ \frac{e^{v z_0} - r z_0}{s (w z_0 - 1) \,  z_0^{j+1}} \ + 
                 \ \frac{e^{v \overline{z_0}} - r \overline{z_0}}{s (w \overline{z_0} - 1) \,  \overline{z_0}^{j+1}} \\
       &=& \ \frac{1}{s}  \   \frac{w e^{v z_0}  \, \overline{z_0}^{j+2} -  e^{v z_0}  \, \overline{z_0}^{j+1}  - r w  \big( a^2 + b^2  \big) \,  \overline{z_0}^{j+1}  
               + r \big( a^2 + b^2  \big) \,  \overline{z_0}^{j}  }{\Big((w a -1)^2 + w^2 b^2 \Big) \, \big( a^2 + b^2  \big)^{j+1}} \\
          &+& \ 
\frac{1}{s}  \   \frac{w e^{v \overline{z_0}}  \, {z_0}^{j+2} -  e^{v \overline{z_0}}  \, {z_0}^{j+1}  - r w  \big( a^2 + b^2  \big) \,  {z_0}^{j+1}  
               + r \big( a^2 + b^2  \big) \,  {z_0}^{j}  }{\Big((w a -1)^2 + w^2 b^2 \Big) \, \big( a^2 + b^2  \big)^{j+1}}.
\end{eqnarray*}
With $e^{v z_0}=e^{va} \, (\cos(v b) + \sin(v b) \, i)$ and
$z_0^j = \rho^j  (\cos(j \phi) + \sin(j \phi) \, i)$ one gets 
\begin{eqnarray*}
\mbox{Re}\big(e^{v {z_0}}  \, \overline{z_0}^{j+2} \big) \ &=& \ \mbox{Re}\big(e^{va} (\cos(vb)+\sin(vb) \, i) \, \rho^{j+2}(\cos((j+2)\phi )
    - \sin((j+2)\phi) \, i \big)   \\
       &=& \ e^{va} \rho^{j+2} \big( \cos(vb) \cos((j+2)\phi) + \sin(vb) \sin((j+2)\phi)     \big)  \\
       &=& \ e^{va} \rho^{j+2} \cos(vb - (j+2)\phi),
\end{eqnarray*}
and substituting $a^2+b^2=\rho^2$  and $z + \overline{z} = 2 \mbox{Re}(z)$ shows that $\mu + \overline{\mu}$ equals
\begin{eqnarray*}
 \frac{2}{s}  \   \frac{w e^{va} \rho^{j+2}\cos(vb-(j+2)\phi)
   - e^{va}\rho^{j+1}\cos(vb-(j+1)\phi)-rw\rho^{j+3}\cos((j+1)\phi) + r\rho^{j+2}\cos(j\phi)}{\big((w a -1)^2 + w^2 b^2 \big) \, \rho^{2{j+2}}}.
\end{eqnarray*}
\end{proof}

The sign of $- (\mu + \overline{\mu})$ in Lemma~\ref{ressum-lem} is determined by the four cosine terms.  If we substitute $j$ with a real ``time'' variable $t$, we can consider them as sine waves of the same frequency, $\phi$, but different amplitudes and phases. A finite sum of such waves is again a sine wave of frequency $\phi$, so that
\begin{eqnarray}
   \Big(\  \frac{e^{va}}{\rho^2}\, \cos((t+1)\phi - vb) \, &-& \, \frac{e^{va}w}{\rho}\, 
                                                                \cos((t+2)\phi -vb)   \label{strongterm} \\
                                                          - \ \big(  \ \frac{r}{\rho} \, \cos(t \phi) \, &-&  \ rw \, \cos((t+1)\phi)   \ \big) \  \Big) 
                                                                        \label{weakterm}\\
                                                                &=& \ L \, \sin(t\phi + p)  \label{ampphase}
\end{eqnarray}
holds, with some amplitude $L$ and some phase $p$. 

\begin{lemma}   \label{sine-lem}
We have
\begin{eqnarray}
   L \ \ge \ L_0 \, := \, \sqrt{w^2 \rho^2 \, - \, 2w \rho \cos(\phi) +1} \, \Big( \frac{e^{va}}{\rho^2} \, - \, \frac{r}{\rho}   \Big).    \label{L0}
\end{eqnarray}
\end{lemma}

\begin{proof}
In general, one has
\[
   a_1 \, \sin(t \phi + p_1) \ + \  a_2 \, \sin(t \phi + p_2) \ = \ \sqrt{a^2_1 + a^2_2 + 2 a_1 a_2 \cos(p_1-p_2)} \, \sin(t \phi +p),
\]
where phase $p$ depends on $a_1,a_2,p_1,p_2$.\, These formulae for the sum of two waves of identical frequency can be found
in textbooks or, for example, in \emph{Bronstein et al., Taschenbuch der Mathematik, 1993}. 
This yields
\begin{eqnarray}
    \frac{e^{va}}{\rho^2}\, \cos((t &+& 1)\phi - vb) \, - \, \frac{e^{va}w}{\rho}\,  \cos((t+2)\phi -vb)  \\
         &=& \, \frac{e^{va}}{\rho^2}\, \sin(t\phi +\phi- vb +\frac{\pi}{2}) \, + \, \frac{e^{va}w}{\rho}\,  \sin(t\phi + 2\phi-vb + \frac{3\pi}{2})  \\
         &=& \, \sqrt{\frac{e^{2va}}{\rho^4}   + \frac{e^{2va}w^2}{\rho^2}   + 2 \frac{e^{2va}w}{\rho^3} \cos(-\phi-\pi)  } \, \sin(t\phi + p) \\
        &=& \, \frac{e^{va}}{\rho^2} \sqrt{w^2 \rho^2 \, -  \, 2w \rho \, \cos(\phi) +1} \, \sin(t\phi + p).   \label{strong}
\end{eqnarray}
Similarly,
\begin{eqnarray}
 - \ \big( \, \frac{r}{\rho} \, \cos(t \phi) \, &-&  \ rw \, \cos((t+1)\phi)  \, \big) \\
    &=& \, \frac{r}{\rho} \sqrt{w^2 \rho^2 \, -  \, 2w \rho \, \cos(\phi) +1} \, \sin(t\phi + q).   \label{weak}
\end{eqnarray}
Thus, the sum of these two sine waves has amplitude
\begin{eqnarray}
\sqrt{w^2 \rho^2 \, -  \, 2w \rho \, \cos(\phi) +1} \, \sqrt{\frac{e^{2va}}{\rho^4}  + \frac{r^2}{\rho^2} + 
       2 \frac{e^{va}}{\rho^2} \frac{r}{\rho} \, \cos(p - q)     }      \label{sumamp} \\
    \geq    \sqrt{w^2 \rho^2 \, -  \, 2w \rho \, \cos(\phi) +1} \,           \Big( \frac{e^{va}}{\rho^2} \, - \,  \frac{r}{\rho}   \Big).
\end{eqnarray}

\end{proof}

Now we can prove a quantitative version of Theorem~\ref{qual-theo} showing that at speed $v>v_c$ the fire fighter is successfull after at most $O(1/\phi)$ many rounds,
where $\phi$ denotes the complex argument of the smallest zero of $e^{wZ}-sZ$.
\begin{theorem}      \label{quant-theo}
Let $\alpha > \alpha_c$.
Then there is an index $j \in  O(\frac{1}{\phi}) $ such that $F_j < 0$ holds.
\end{theorem}
\begin{proof}
The function $h(t) := L \, \sin(t\phi + p)$ of~\ref{ampphase} attains its minima $-L$ at
arguments $t^*$ where  $t^* \phi + p \, \equiv \,  \frac{3 \pi}{2} \mod 2\pi$. For an integer $j$
at most $1/2$ away from~$t^*$ we have
\[
     h(j) \, \leq \,  h(t^*+\frac{1}{2}) \,  = \,  L \, \sin(\frac{3\pi+\phi}{2}) \, = \, - \, L \, \cos(\frac{\phi}{2});
\]
these terms are negative because of $\phi < 2.09< \pi$. This implies
\begin{eqnarray}
\frac{F_j}{F_0} \ < \ - \ \frac{2}{s}  \, \Big( \frac{e^{va}}{\abs{z_0}} \, - \, r \Big) \, \frac{1}{w} \, \frac{1}{\abs{z_0-x_0}} \,  \cos(\frac{\phi}{2}) \, \abs{z_0}^{-j} \ + \ D \, \cdot \,  0.9^{-j},    \label{upbound}
\end{eqnarray}
summarizing~\ref{CRT}, \ref{major} to~\ref{power}, \ref{ampphase}, Lemma~\ref{sine-lem}, and~\ref{int}.
 We observe that such integers~$j$ occur (at least) once in every period of length $2\pi/\phi$ of function $h(t)$. 

Since $\abs{z_0} < 0.31 < 0.9$, the powers $\abs{z_0}^{-j}$ grow in $j$ much faster than $0.9^{-j}$ does.
All coefficients in~\ref{upbound} are positive, and lower bounded by independent constants on 
$[\alpha_c, \pi/2]$.
Indeed, we have $\frac{2}{s} \geq 0.091$ with a minimum at $\alpha_c$, 
$\frac{e^{va}}{\abs{z_0}}  -  r \geq 1.581$ with its minimum at $\alpha=\pi/2$, and
$w \abs{z_0-x_0} \leq 0.33$ with a maximum at $\alpha=\pi/2$.
 
Hence, after a constant number of periods the value of~\ref{upbound} becomes negative. This completes the proof of Theorem~\ref{quant-theo}.
\end{proof}
Numerical inspection shows that a suitable integer $j$ can already be found in the first period of function $h$, so that $j \leq \frac{2\pi}{\phi} +1$. 

\vspace{\baselineskip}
Now we let speed $v$ decrease to the critical value $v_c$, and prove that the first index~$j$,
for which $F_j$ becomes negative, grows with $\pi/\phi$ to infinity. To this end we prove
that the sine wave in~\ref{ampphase} starts, at zero, near the beginning of its positive half-cycle, so that it takes half a period before negative values can occur.

The graph of $\sin(t \phi +p)$ is shifted, along the $t$-axis, by $p/\phi$ to the left, as compared to
the graph of $\sin(t)$. As $\alpha$ tends to $\alpha_c$, frequency $\phi$ goes to zero, and so does
phase $p$.  But, surprisingly, their ratio rapidly converges to a small constant.

\begin{lemma}      \label{shift-lem}
We have
\begin{eqnarray*}
  \sigma \ := \  \lim_{\alpha \rightarrow \alpha_c} \ \frac{p}{\phi} \ =  \ \Big(\frac{r}{e^{\frac{v}{w}} w -r} +1\Big) \,           
\big(1-\frac{v}{w}\big) \ + \ \frac{1}{3} 
      \approx \ 1.351.
 \end{eqnarray*}
\end{lemma}
Figure~\ref{wave-fig} shows $-(\mu + \overline{\mu}) \rho^{j-1}$ as a
function of time parameter $j=t$; see ~\ref{major} to~\ref{power}.
\begin{figure}
\begin{center}
\includegraphics[scale=0.4]{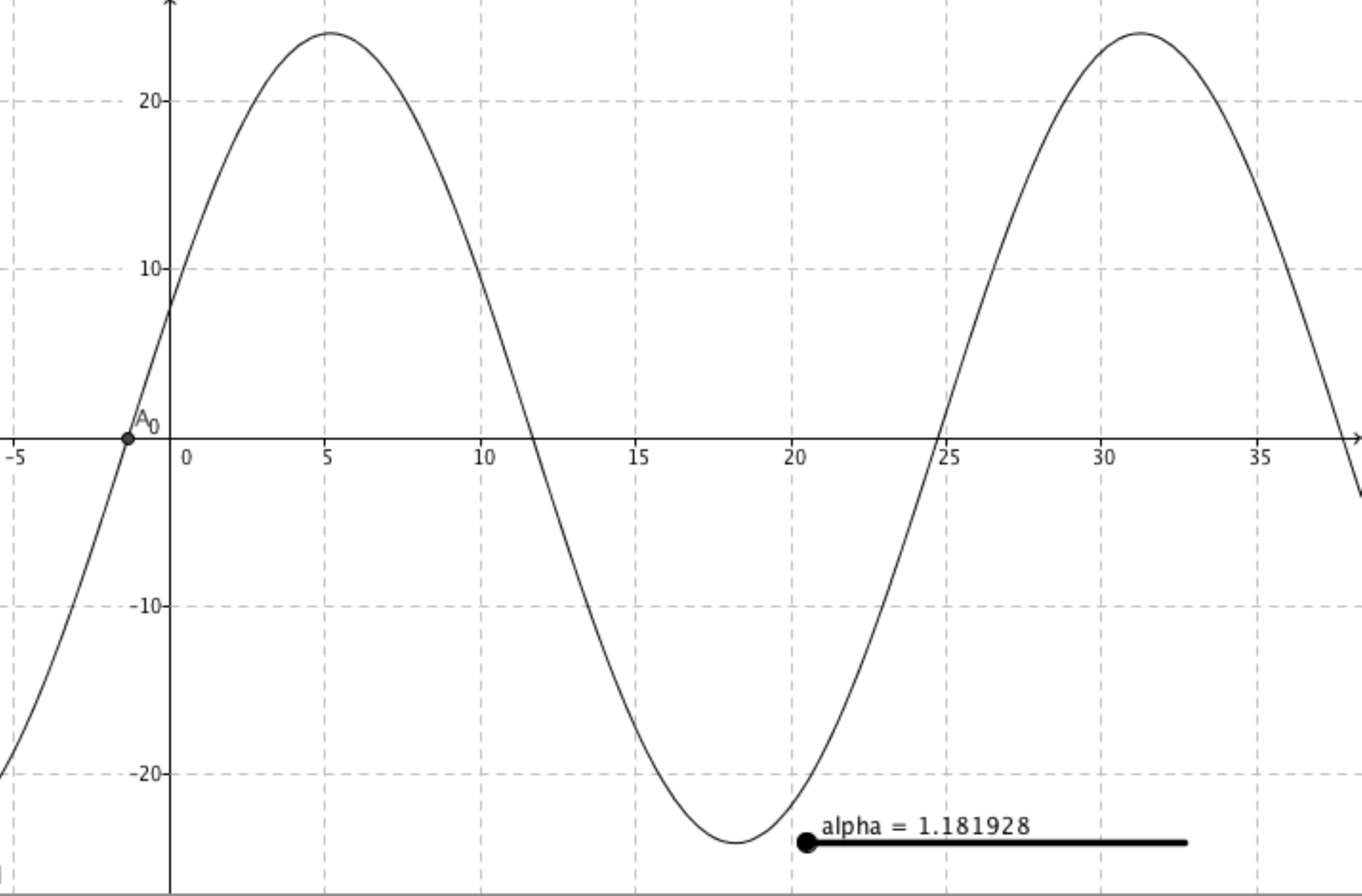}
\caption{The shift to the left is almost constant, as $\alpha$ tends to $\alpha_c$ 
and period $2 \pi/\phi$ goes to infinity.}
\label{wave-fig}
\end{center}
\end{figure}
Crucial in the proof is the following geometric fact.
\begin{lemma}      \label{tria-lem}
Consider the triangle shown in Figure~\ref{fulltria-fig}, which has a base of length~1, a base angle
of $\phi$, and height $\phi$. As $\phi$ goes to zero, 
the ratio $\frac{\tau}{\phi}$ tends to $1/3$, and $\gamma$ converges to $\pi/2$.
\end{lemma}

\begin{figure}
\begin{center}
\includegraphics[scale=0.5]{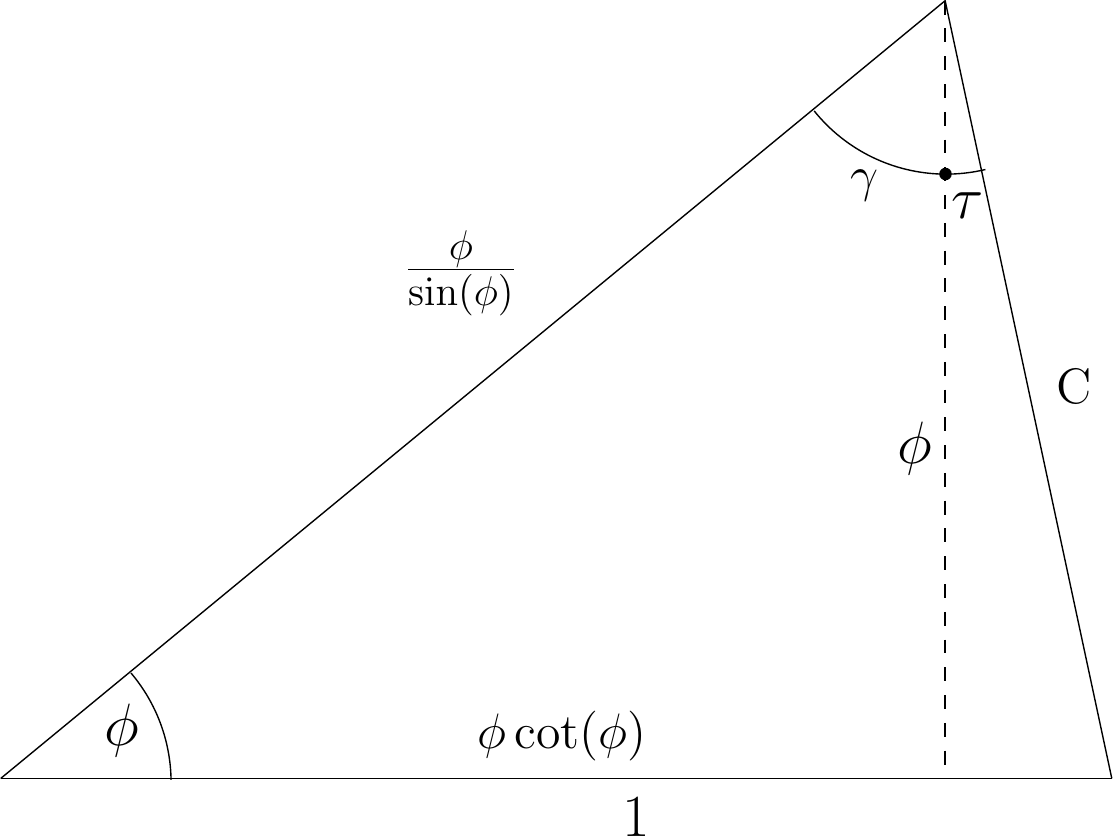}
\caption{Ratio $\tau / \phi$ tends to $1/3$, as $\phi$ goes to zero.}
\label{fulltria-fig}
\end{center}
\end{figure}
\begin{proof}
From $C \cos\tau = \phi$ and
$C \sin\tau = 1-\phi \cot\phi$ we obtain
\[
    \tan\tau \, = \, \frac{\sin\phi - \phi \cos\phi}{\phi \sin\phi},
\]
and because $\tau$ must go to zero as $\phi$ does, we have
\[
     \frac{\tau}{\phi} \, = \, \cos\tau \, \frac{\tau}{\sin\tau} \,  \frac{\tan\tau}{\phi}
         \, \sim \, \frac{\sin\phi - \phi \cos\phi}{\phi^2 \sin\phi}.
\]
A twofold application of  l'Hospi\-tal's rule shows that the last term has the same limit as
\[
   \frac{\sin\phi}{2 \sin\phi + \phi \cos\phi} \, \sim \, \frac{\cos\phi}{3 \cos\phi - \phi \sin\phi},
\]
which converges to $1/3$.

Moreover, we have 
\[
   \sin\gamma \, = \, \sin(\pi/2 - \phi) \, = \, \cos(\phi)
\]
which tends to~1, so that $\gamma$ converges to $\pi/2$.

\end{proof}

Now we give the proof of Lemma~\ref{shift-lem}.
\begin{proof}
As in the proof of Lemma~\ref{sine-lem} one generally has
\[
   a_1 \, \sin(t \phi + p_1) \ + \  a_2 \, \sin(t \phi + p_2) \ = \  a_3 \, \sin(t \phi +p_3),
\]
where the new amplitude, $a_3$, is given by
\[
    a_3 \ = \ \sqrt{a^2_1 + a^2_2 + 2 a_1 a_2 \cos(p_1-p_2)},
\]
and the new phase, $p_3$, fulfills
\[
  p_3 \ = \   \arcsin(\frac{a_2 \, \sin(p_2-p_1)}{a_3}) \ + \ p_1.
\]
First, we are applying this formula to~\ref{strongterm}, 
\begin{eqnarray*}
   \frac{e^{va}}{\rho^2}\, \cos((t+1)\phi - vb) \, &-& \, \frac{e^{va}w}{\rho}\, 
                                                                \cos((t+2)\phi -vb)    \\
    = \ \frac{e^{va}}{\rho^2}\, \sin(t\phi + \phi - vb + \pi/2) \, &+& \, \frac{e^{va}w}{\rho}\, 
                                                                \sin(t\phi + 2\phi -vb + 3\pi/2) 
\end{eqnarray*}
and obtain, as in~\ref{strong},
\[
   a_{\ref{strongterm}} \ = \ \frac{e^{va}}{\rho^2} \sqrt{w^2 \rho^2 \, -  \, 2w \rho \, \cos(\phi) +1},
\]
and, for the new phase,
\begin{eqnarray*}
   p_{\ref{strongterm}} \ &=&  \ \arcsin(-\frac{w \rho \sin\phi}{\sqrt{w^2 \rho^2 - 2w \rho \cos\phi +1}}) \ + \ \phi - vb + \pi/2  \\
     &=& \  \arcsin(-\frac{\phi}{\sqrt{\frac{\phi^2}{\sin^2\phi} -2\phi \frac{\cos\phi}{\sin\phi} +1}}) \ + \ \phi - vb + \pi/2  \\
     &=& \  \arcsin( - \cos\tau)  \ + \ \phi - vb + \pi/2   \\
     &=& \  \arcsin(\sin(-\pi/2 + \tau))  \ + \ \phi - vb + \pi/2  \\
     &=& \   \tau \, + \, \phi - vb
\end{eqnarray*}
using~\ref{philab} and the triangle in Figure~\ref{fulltria-fig}.
We conclude that the value of $\arcsin$ goes to $-\pi/2$, so that $p_{\ref{strongterm}}$ converges to zero.
For the resulting shift we obtain
\begin{eqnarray*}
\frac{p_{\ref{strongterm}}}{\phi} \ &=& 
       \ \frac{\tau}{\phi} \ + \ 1 \ - \frac{v}{w}   \\
       &\rightarrow& \ \frac{1}{3} \, + \, 1 - \frac{v}{w}
\end{eqnarray*}
by Lemma~\ref{tria-lem}.

\vspace{\baselineskip}
Next, we consider~\ref{weakterm},
\begin{eqnarray*}
 \ rw \, \cos((t+1)\phi) \,  &-&  \,   \frac{r}{\rho} \, \cos(t \phi)  \\
     &=& \ rw \, \sin(t \phi +\phi + \pi/2) \, + \, \frac{r}{\rho} \, \sin(t \phi + 3\pi/2)
\end{eqnarray*}
As in~\ref{weak},
\[
     a_{\ref{weakterm}} \    = \, \frac{r}{\rho} \sqrt{w^2 \rho^2 \, -  \, 2w \rho \, \cos(\phi) +1},
\]
and for the phase,
\begin{eqnarray*}
   p_{\ref{weakterm}} \ &=&  \ \arcsin\Big(\frac{\sin\phi}{\sqrt{w^2 \rho^2 - 2w \rho \cos\phi +1}}\Big) \ + \ \phi  + \pi/2  \\
           &=& \ \pi/2 \, + \, \gamma \, + \, \tau \, + \, \phi,
\end{eqnarray*}
observing that the argument of $\arcsin$ equals
\[
    \frac{\sin\phi}{C} \, = \, \frac{\sin(\gamma+\tau)}{1},
\]
applying the law of sines to the triangle shown in Figure~\ref{fulltria-fig}.
We see that $p_{\ref{weakterm}}$ goes to $\pi$; in this case, the shift does not converge.

\vspace{\baselineskip}
Now we consider the sum of~\ref{strongterm} and \ref{weakterm}. We know from~\ref{sumamp} the final amplitude,
\begin{eqnarray*}
a \ = \ \sqrt{w^2 \rho^2 \, -  \, 2w \rho \, \cos(\phi) +1} \, \sqrt{\frac{e^{2va}}{\rho^4}  + \frac{r^2}{\rho^2} + 
       2 \frac{e^{va}}{\rho^2} \frac{r}{\rho} \, \cos( p_{\ref{weakterm}} -  p_{\ref{strongterm}} )     },
\end{eqnarray*}
and obtain for the phase
\begin{eqnarray*}
 p \ = \  \arcsin\Biggl( \frac{r}{\rho} \, \frac{\sin(p_{\ref{weakterm}} - p_{\ref{strongterm}} )}{\sqrt{\frac{e^{2va}}{\rho^4}  + \frac{r^2}{\rho^2} + 
       2 \frac{e^{va}}{\rho^2} \frac{r}{\rho} \, \cos( p_{\ref{weakterm}} -  p_{\ref{strongterm}}) }} \Biggr) \ + \ p_{\ref{strongterm}}.
\end{eqnarray*}
For short, let $\overline{p}$ denote the $\arcsin$ term, and 
let $R$ be the square root in the denominator.
Since $p_{\ref{weakterm}} - p_{\ref{strongterm}}$ tends to $\pi$ we conclude that $\overline{p}$ goes to zero. Thus, we obtain
\begin{eqnarray*}
    \frac{p}{\phi} \ &=& \ \frac{\overline{p}}{\sin \overline{p}} \, \frac{\sin \overline{p}}{\phi}  \, + \,           \frac{p_{\ref{strongterm}}}{\phi} \\
          \ &=& \ \frac{\overline{p}}{\sin \overline{p}} \,  \frac{r}{\rho} \, \frac{1}{R} \,   \frac{\sin(p_{\ref{weakterm}} 
       - p_{\ref{strongterm}} )}{\pi - (p_{\ref{weakterm}} - p_{\ref{strongterm}} )} \ \cdot \ \frac{\pi - (p_{\ref{weakterm}} - p_{\ref{strongterm}} )}{\phi} \ + \ \frac{p_{\ref{strongterm}} }{\phi}  \\
     &\sim& \  \frac{r}{\rho} \, \frac{1}{R} \ \cdot \ \frac{\pi/2 - \gamma - vb}{\phi} \ + \ \frac{p_{\ref{strongterm}} }{\phi} \\
     &=& \  \frac{r}{\rho} \, \frac{1}{R} \ \cdot \ \frac{\phi - vb}{\phi} \ + \ \frac{p_{\ref{strongterm}} }{\phi} \\
     &\rightarrow& \ \frac{rw}{e^{v/w}w^2 - rw}\, \Big(1-\frac{v}{w}\Big) \ + \ 1/3 + 1 - \frac{v}{w}.
\end{eqnarray*}
This concludes the proof of Lemma~\ref{shift-lem}.
\end{proof}

\vspace{\baselineskip}
Now let $j$ be an integer satisfying
\[
      j \, \leq \, \frac{\pi}{\phi} \, - \, 2 \, \sigma.
\]
Lemma~\ref{shift-lem} implies that the sign of $\sin(j \phi +p)$ in~\ref{ampphase} is positive,
and, even more,  that 
\begin{eqnarray}
   \sin(j \phi +p) \, > \, \sin(p) > \sin(1.35 \, \phi)
\end{eqnarray}
holds. For such integers $j$ we obtain, similarly to~\ref{upbound}, 
\begin{eqnarray}
\frac{F_j}{F_0} \ > \  \frac{2}{s}  \, \Big( \frac{e^{va}}{\abs{z_0}} \, - \, r \Big) \, \frac{1}{w} \,
\frac{1}{\abs{z_0-x_0}} \,  \sin(1.35 \, \phi) \, \abs{z_0}^{-j} \ - \ 0.1 \, \cdot \,  0.9^{-j}    \label{lobound}
\end{eqnarray}
Here we have used the following estimate for $D$ in~\ref{int}.
\begin{lemma}    \label{D-lem}
With the radius of $\Gamma$ equal to $\gamma=0.9$, we have $\abs{f(z)} \leq 0.1$
for all $z$ on $\Gamma$, if $\alpha$ is close enough to $\alpha_c$.
\end{lemma}
\begin{proof}
Let $z=\gamma \, e^{i \psi}$ be a parameterization of circle $\Gamma$ for $\psi \in [0 \ldots 2\pi]$. By multiplication with complex conjugates,
\begin{eqnarray*}
    \abs{f(z)}_{z \in \Gamma} \ &=& \ \abs{\frac{ e^{v \gamma e^{i \psi}} \ - \ r \, \gamma e^{i \psi}}{ e^{w \gamma e^{i \psi}} \ -  
        \ s \, \gamma e^{i \psi}}  } \\
              &=& \sqrt{\frac{  e^{2v \gamma \cos\psi} - 2 e^{v \gamma \cos\psi} \, r \gamma \, \cos(v \gamma \sin\psi - \psi)  + r^2 \gamma^2}             {e^{2w \gamma \cos\psi} - 2 e^{w \gamma \cos\psi} \, s \gamma \, \cos(w \gamma \sin\psi - \psi)  + s^2 \gamma^2}                      }.
\end{eqnarray*}
The maximum is attained at $\psi=\pi$, and it grows monotonically from $0.09\ldots$ for $\alpha=\alpha_c$ to  $1.269\ldots$ for 
$\alpha = \pi/2$.
\end{proof}

Now we can state the lower bound.
\begin{theorem}    \label{lowbo-theo}
As angle $\alpha$ decreases to the critical value $\alpha_c$, the number $j$ of rounds necessary to contain 
the fire is at least $j > \frac{\pi}{\phi} - 2.71$. This lower bound grows to infinity.
\end{theorem}
\begin{proof}
By the preceeding discussion, estimate~\ref{lobound} holds for each $j$ that stays below this bound.
As $\phi$ tends to~0 we get
\begin{eqnarray*}
\frac{1}{w} \, \frac{1}{\abs{z_0-x_0}} \,  \sin(1.35 \phi) \ &=& \ 
       \frac{\sin(1.35 \phi)}{\sqrt{w^2 \rho^2 -2w \rho \cos(\phi) +1}}  \\
          &=& \ \frac{\sin(1.35 \phi)}{\sqrt{\frac{\phi^2}{\sin^2\phi} -2\phi \frac{\cos\phi}{\sin\phi} +1}} \\
          &\sim& \ 1.35 \, \frac{\sin\phi}{\sqrt{\frac{\phi^2}{\sin\phi^2} -2\phi \frac{\cos\phi}{\sin\phi} +1}} \\
          &=& \ 1.35 \, \sin(\gamma + \tau) \\
          &\sim& \, 1.35.
\end{eqnarray*}
Here, the first equality follows from~\ref{wurz} and the second, from~\ref{philab}. Then we have applied
l'Hospi\-tal's rule, and the next line follows from
Lemma~\ref{tria-lem}. Indeed, the square root is equal to $C$ in Figure~\ref{fulltria-fig}, and we can apply
the law of sines together with the fact that $\gamma$ goes to $\pi/2$, and $\tau$ to zero.

Substituting in~\ref{lobound} the other limit values (non of which is critical) we find
\begin{eqnarray*}
\frac{F_j}{F_0} \ &>& \  0.091 \,\cdot \, 7.82 \, \cdot \, 1.35 \, \cdot \, \abs{z_0}^{-j} \ - \ 0.1 \, \cdot \,  0.9^{-j}  \\
                         &\geq& \ 0.96 \, \cdot \, 0.1239^{-j} \ - \ 0.1 \, \cdot \,  0.9^{-j}   \\
                         &>& 0.
\end{eqnarray*}
Here 7.82 is the limit of $\frac{e^{va}}{\abs{z_0}} -r$ as $\alpha$ tends to $\alpha_c$. This completes the proof of 
Theorem~\ref{lowbo-theo}.
\end{proof}

Figure~\ref{j-fig} shows how many rounds the fighter needs to contain the fire, depending on her speed $v$.
\begin{figure}
\begin{center}
\includegraphics[scale=0.3]{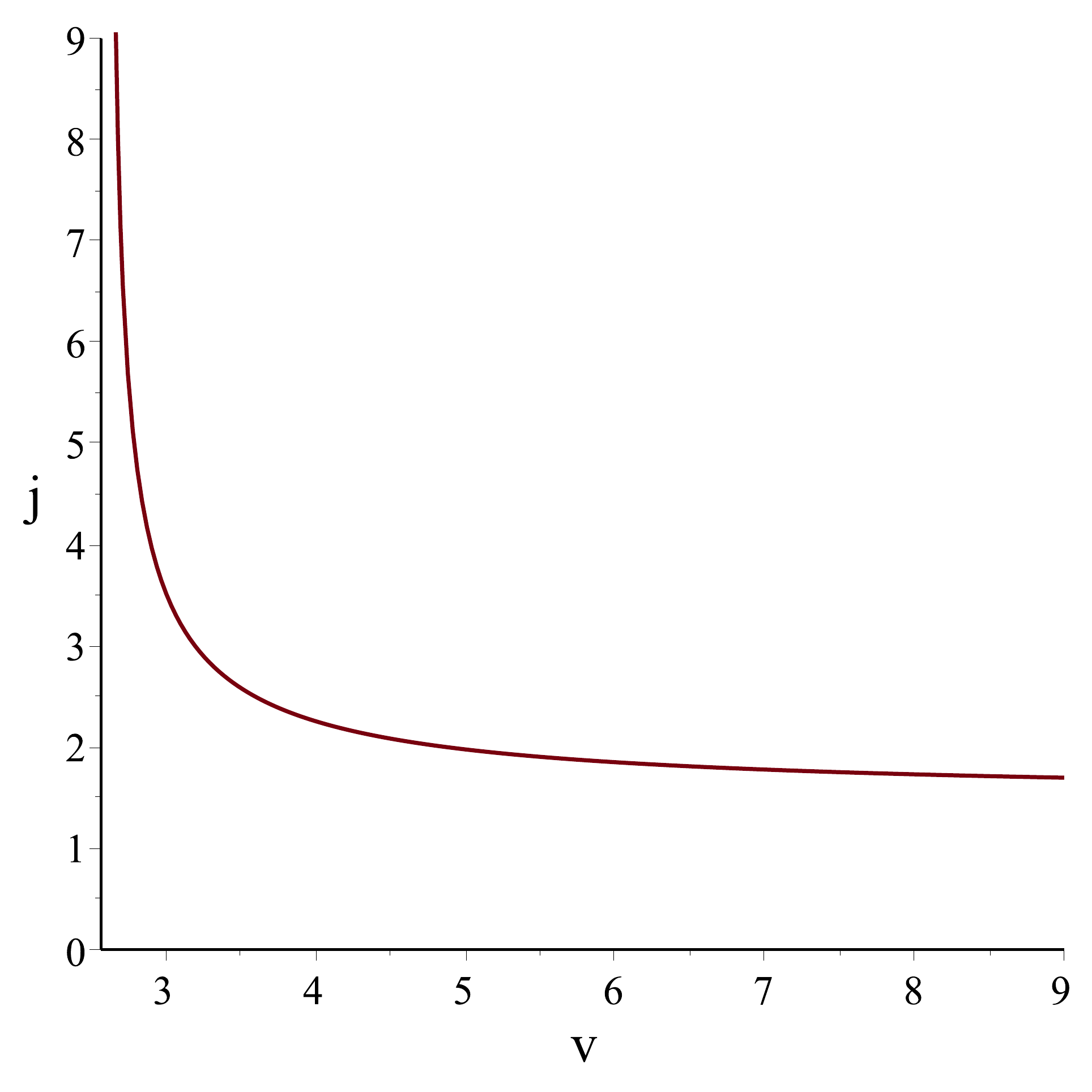}
\caption{The approximate number of rounds, $j$, barrier curve $\mbox{FF}_v$ needs before closing on itself.} 
\label{j-fig}
\end{center}
\end{figure}
%

\section{Lower bound}\label{lobowrap-sec}

In this section a barrier curve $S$ is called {\em spiralling} if it starts on the boundary of a fire of radius $A$, and visits the four coordinate half-axes in counterclockwise order and at increasing distances from the origin.  We are proving the following.

\begin{theorem}   \label{lowbound-theo}
In order to contain a fire by a spiralling barrier, the fighter needs speed 
\[
v \ > \ \frac{1+\sqrt{5}}{2} \, \approx \, 1.618,
\]
the golden ratio.
\end{theorem}
\begin{proof}
Now let $S$ be a spiralling curve, and assume that the fighter proceeds
at maximum speed $v \leq (1 +\sqrt 5)/2$.
Let $p_0, p_1, p_2, \ldots$ denote the points on the coordinate axes visited, in this order, by $S$.
The following lemma shows that $S$ cannot succeed because there is still fire burning outside the barrier on the axis previously visited.
\begin{figure}
\begin{center}
\includegraphics[scale=0.6]{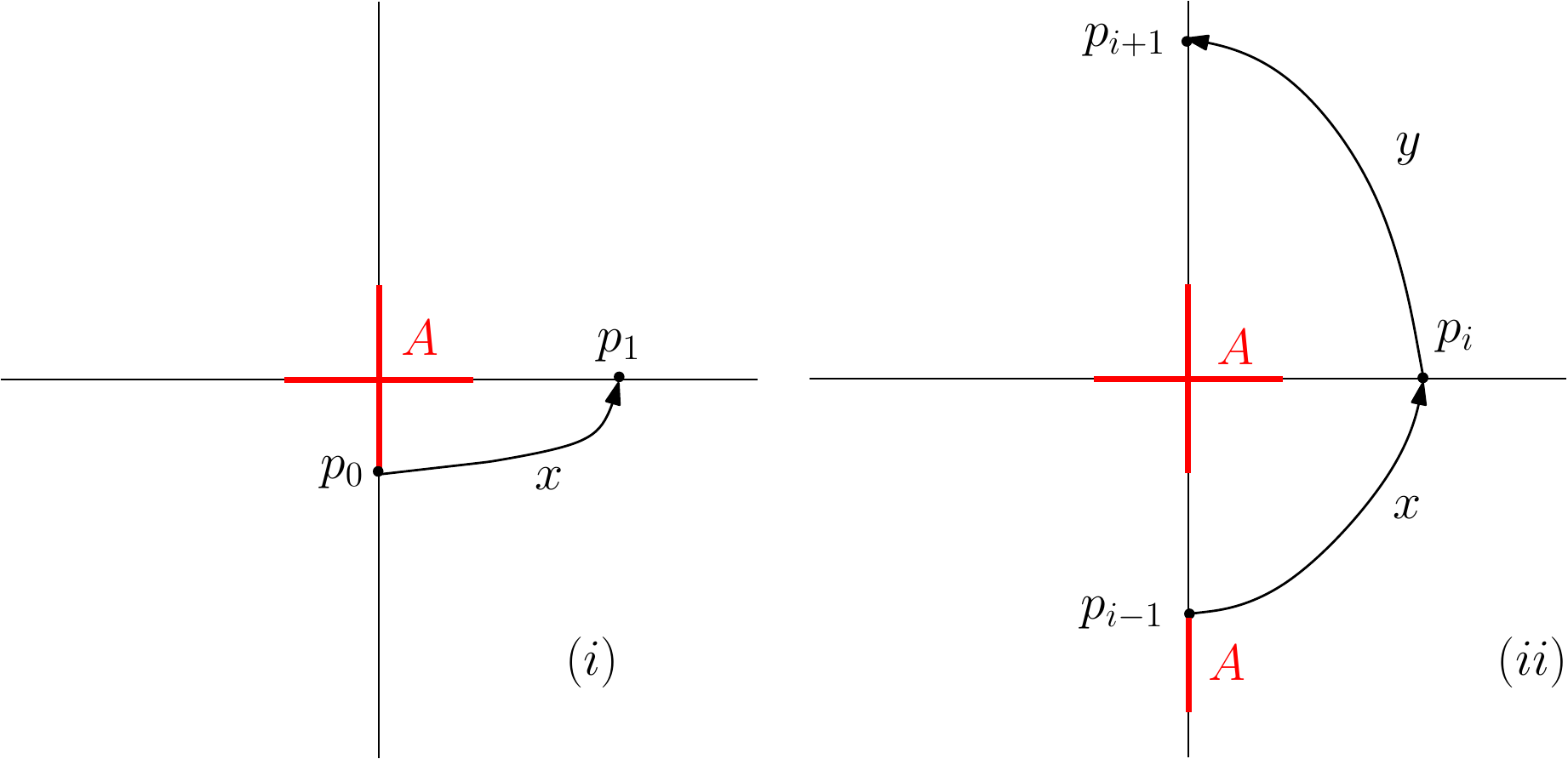}
\caption{Proof of Lemma~\ref{inv-lem}.}
\label{lowbo-fig}
\end{center}
\end{figure}
\begin{lemma}     \label{inv-lem}
Let $A$ be the initial fire radius. When $S$ visits point $p_{i+1}$, the interval $[p_i,p_i+{\tt sign}(p_i)A]$ on the axis visited before is on fire.
\end{lemma}
\begin{proof}
The proof is by induction on $i$. Suppose barrier $S$ is of length $x$ between $p_0$ and $p_1$, as shown 
in Figure~\ref{lowbo-fig}~(i). While this part is under construction, the fire advances $x/v$ along the positive $X$-axis, so that
$A+x/v \leq p_1 \leq x$ must hold, or 
\[
     \frac{x}{v} \, \geq \, \frac{1}{v-1} A \, > \, A;
\]
the last inequality follows from $v<2$. Thus, the fire has enough time to move a distance of $A$ from $p_0$ downwards along
the negative $Y$-axis.

Now let us assume that the fighter builds a barrier of length $y$ between $p_i$ and $p_{i+1}$,  as shown 
in Figure~\ref{lowbo-fig}~(ii). By induction, the interval of length $A$ below $p_{i-1}$ is on fire. Also, when the fighter moves
on from $p_i$, there must be a burning interval of length at least $A+x/v$ on the positive $Y$-axis which is not bounded by a
barrier from above. This is clear if $p_{i+1}$ is the first point visited on the positive $Y$-axis, and it follows by induction, otherwise. Thus, we must have $A+x/v +y/v \leq p_{i+1} \leq y$, hence
\[
    \frac{y}{v}  \, \geq \, \frac{1}{v-1} A \,  + \, \frac{1}{v(v-1)} x \,  > \, A+x.
\]
The rightmost inequality follows since $v$ is supposed to be smaller than the golden ratio, which
satisfies $X^2 -X -1 =0$; hence, $v^2-v <1$. 
This shows that the fire has time to crawl along the barrier
from $p_{i-1}$ to $p_i$, and a distance $A$ to the right, as the fighter moves to $p_{i+1}$, completing the proof of Lemma~\ref{inv-lem} and of Theorem~\ref{lowbound-theo}.
\end{proof}

\end{proof}

\end{document}